\documentclass{article}

\usepackage[pdftex]{hyperref}
\usepackage{graphicx}
\usepackage{setspace}
\usepackage{amsmath}
\usepackage{amssymb,amsfonts,amsbsy}
\usepackage{amsthm}
\usepackage{url}

\newcommand{\sums}{\sum_{k=s-N^{1/3}}^{s+N^{1/3}}}
\newcommand{\tm}{t_{\text{max}}}

\newcommand{\vs}{\vert s\vert}
\newcommand{\Om}{\Omega}
\newcommand{\bOm}{\bar{\Omega}}
\newcommand{\X}{\mathfrak{X}}

\newcommand{\Z}{\mathbb{Z}}
\newcommand{\R}{\mathbb{R}}
\newcommand{\C}{\mathbb{C}}
\newcommand{\dn}{\delta_n}

\newcommand{\bigO}{\mathcal{O}}
\newcommand{\lf}{\lfloor}
\newcommand{\rf}{\rfloor}
\newcommand{\E}{\mathbb{E}}

\theoremstyle{plain}
\newtheorem{theorem}{Theorem}[section]
\newtheorem{proposition}[theorem]{Proposition}
\newtheorem{lemma}[theorem]{Lemma}
\newtheorem{corollary}[theorem]{Corollary}

\theoremstyle{definition}
\newtheorem{definition}[theorem]{Definition}

\theoremstyle{remark}
\newtheorem{remark}[theorem]{Remark}

\begin{document}

\title{The Gaussian free field in interlacing particle systems}
\author{Jeffrey Kuan}
\date{}
\maketitle
\begin{abstract} We show that if an interlacing particle system in a two-dimensional lattice is a determinantal point process, and the correlation kernel can be expressed as a double integral with certain technical assumptions, then the moments of the fluctuations of the height function converge to that of the Gaussian free field. In particular, this shows that a previously studied random surface growth model with a reflecting wall has Gaussian free field fluctuations. 
\end{abstract}

\section{Introduction}We begin by describing a particle system which was introduced in \cite{kn:BK1}.

\textbf{Particle System.} Introduce coordinates on the plane as shown in Figure \ref{configur}. Denote the
horizontal coordinates of all particles with vertical coordinate $m$
by $y^m_1>y^m_2>\dots>y^m_k$, where $k=\lfloor (m+1)/2\rfloor$. There is a wall on the left side, which forces $y^m_k\geq 0$ for $m$ odd and $y^m_k\geq 1$ for $m$ even. The particles must also satisfy the interlacing conditions $y_{k+1}^{m+1}<y_k^m<y_k^{m+1}$ for all meaningful values of $k$ and $m$.

\begin{center}
\begin{figure}[htp]
\caption{}
\begin{center}\includegraphics[height=4.5cm]{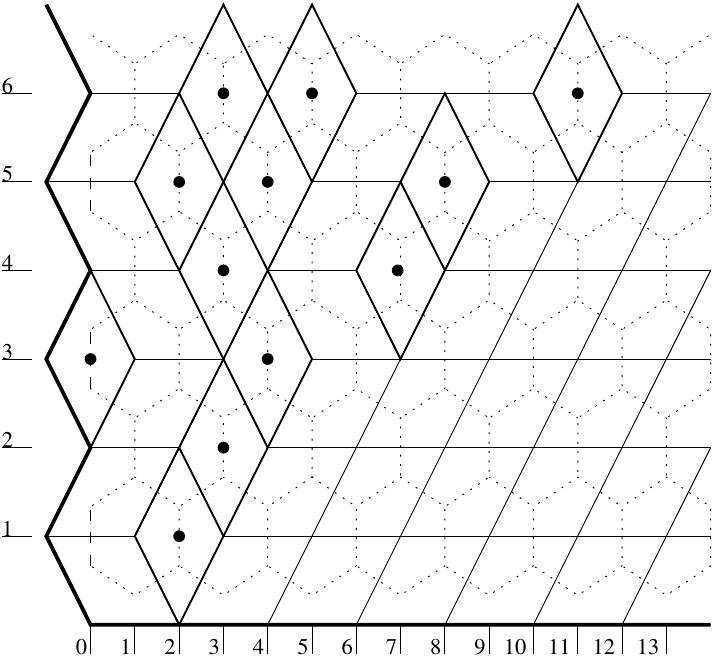}\end{center}
\label{configur}
\end{figure}
\end{center}

By visually observing Figure \ref{configur}, one can see that the particle system can be interpreted as a stepped surface. We thus define the height function at a point to be the number of particles to the right of that point.

Define a continuous time Markov chain as follows. The initial condition is a single particle configuration where all
the particles are as much to the left as possible, i.e.
$y^m_k=m-2k+1$ for all $k,m$. This is illustrated in the left-most iamge in Figure \ref{jumps}. Now let us describe the evolution.
We say that a particle $y^m_k$ is blocked on the right if
$y^m_k+1=y^{m-1}_{k-1}$, and it is blocked on the left if
$y^m_k-1=y^{m-1}_k$ (if the corresponding particle $y^{m-1}_{k-1}$
or $y^{m-1}_{k}$ does not exist, then $y^m_k$ is not blocked).

Each particle has two exponential clocks of rate $\frac 12$; all
clocks are independent. One clock is responsible for the right
jumps, while the other is responsible for the left jumps. When the
clock rings, the particle tries to jump by 1 in the corresponding
direction. If the particle is blocked, then it stays still. If the
particle is against the wall (i.e. $y^m_{[\frac {m+1}2]}=0$) and the
left jump clock rings, the particle is reflected, and it tries to
jump to the right instead.

When $y^m_k$ tries to jump to the right (and is not blocked on
the right), we find the largest $r\in \Z_{\ge 0}\sqcup\{+\infty\}$
such that $y_k^{m+i}=y_k^{m}+i$ for $0\le i\le r$, and the jump
consists of all particles $\bigl\{y_k^{m+i}\bigr\}_{i=0}^r$ moving
to the right by 1. Similarly, when $y^m_k$ tries to jump to the left
(and is not blocked on the left), we find the largest $l\in \Z_{\ge
0}\sqcup\{+\infty\}$ such that $y_{k+j}^{m+j}=y_k^m-j$ for $0\le
j\le l$, and the jump consists of all particles
$\bigl\{y_{k+j}^{m+j}\bigr\}_{j=0}^l$ moving to the left by 1.

In other words, the particles with smaller upper indices can be
thought of as heavier than those with larger upper indices, and the
heavier particles block and push the lighter ones so that the
interlacing conditions are preserved.

\begin{center}
\begin{figure}[htp]
\caption{First three jumps} \medskip
\includegraphics[height=2.8cm]{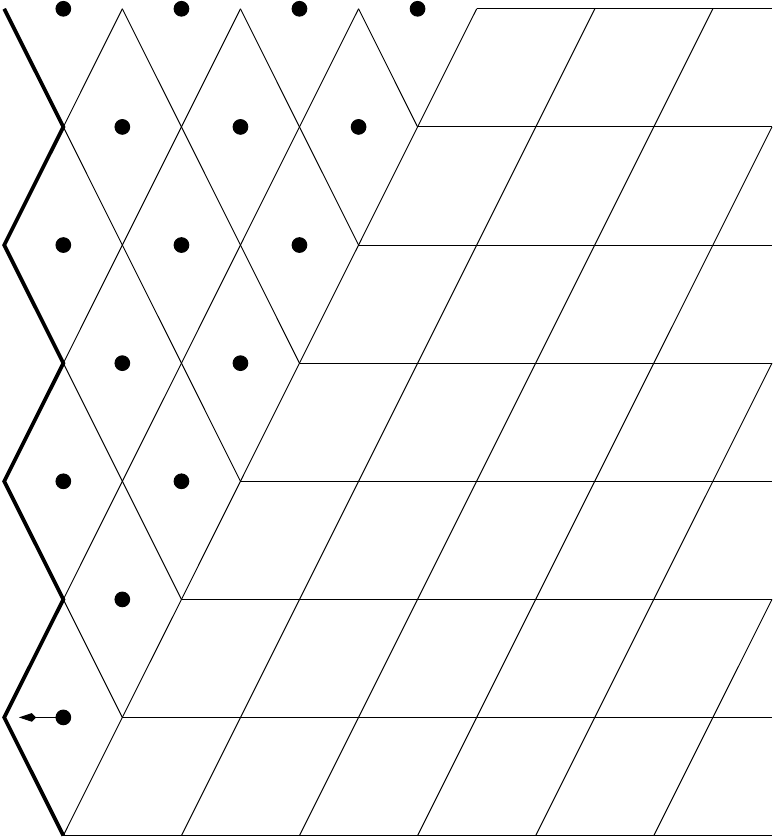}
\quad
\includegraphics[height=2.8cm]{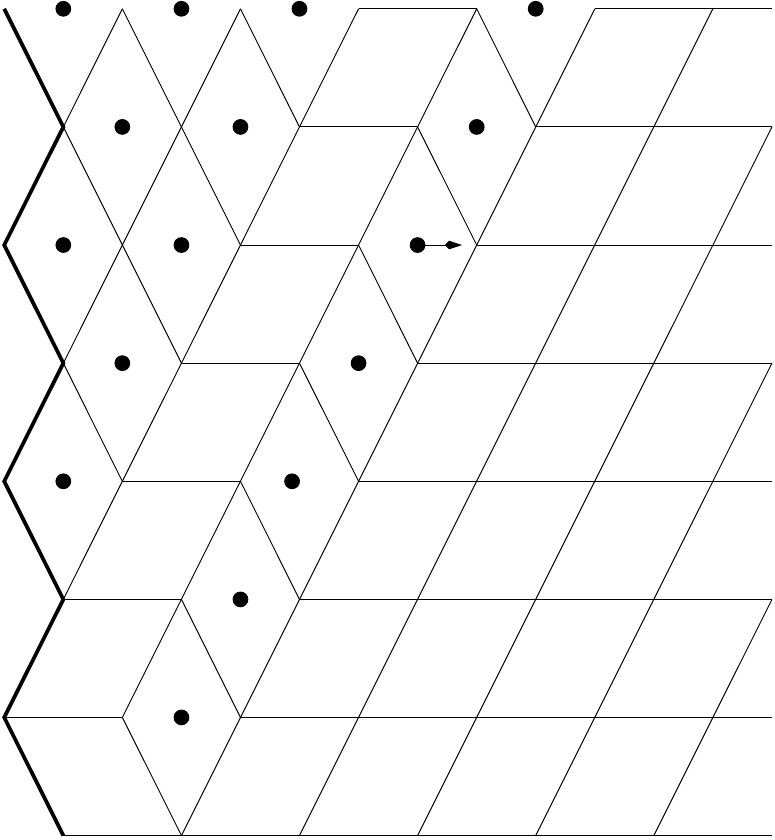}
\quad
\includegraphics[height=2.8cm]{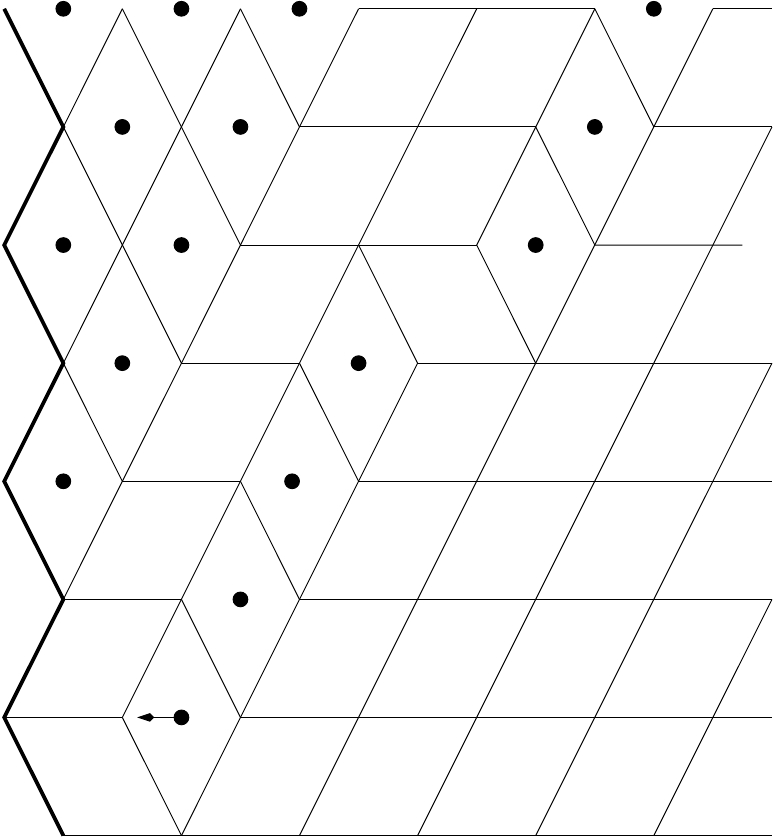}
\quad
\includegraphics[height=2.8cm]{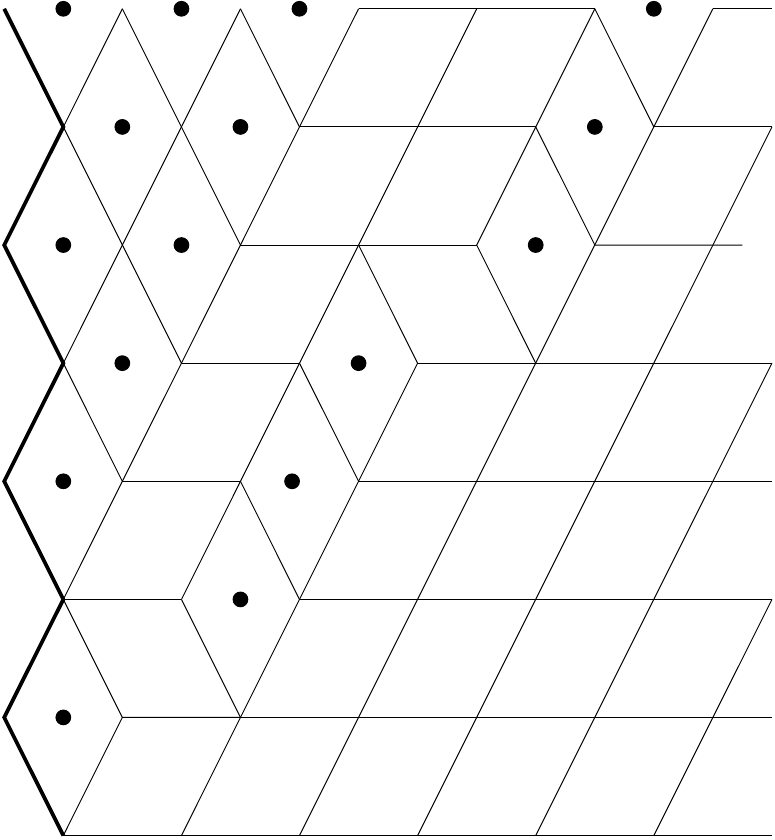} \label{jumps}
\end{figure}
\end{center}
Figure \ref{jumps} depicts three possible first jumps: Left clock of
$y_1^1$ rings first (it gets reflected by the wall), then right
clock of $y_1^5$ rings, and then left clock of $y_1^1$ again.


In terms of the underlying stepped surface, the evolution can be
described by saying that we add possible ``sticks'' with base
$1\times 1$ and arbitrary length of a fixed orientation with rate
1/2, remove possible ``sticks'' with base $1\times 1$ and a
different orientation with rate 1/2, and the rate of removing sticks
that touch the left border is doubled.\footnote{This phrase is based
on the convention that
\includegraphics[height=0.4cm]{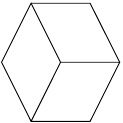} is a figure of a
$1\times1\times1$ cube. If one uses the dual convention that this is
a cube-shaped hole then the orientations of the sticks to be added
and removed have to be interchanged, and the tiling representations
of the sticks change as well.} 

A computer simulation of this dynamics can be found at\\
\href{http://www.math.caltech.edu/papers/Orth_Planch.html}
{$\mathtt{http://www.math.caltech.edu/papers/Orth\_Planch.html}$}.

This particle system has important connections to the representation theory of the orthogonal groups, to the Kardar--Parisi--Zhang equation from mathematical physics, and to random lozenge tilings. The interested eader is referred to the introduction of \cite{kn:BK1}. 

\smallskip

\textbf{Limit shape}
A very natural question about this random surface is to ask if it satisfies a law of large numbers and central limit theorem. In other words, in the large $N$ limit, the random surface should converge to a deterministic limit shape, and the fluctuations around this limit shape should be a reasonably nice object. This paper will prove that the flucutations are described by the Gaussian free field, but first let us describe the limit shape, which was proved in Proposition 5.6 of \cite{kn:BK1}. 

Let $H(x,n,t)$ denote the height function, i.e. the number of particles to the right of $(x,n)$ at time $t$. Define $h$ to be
$$
h(\nu,\eta,\tau) := \lim_{N\rightarrow\infty} \frac{1}{N} \mathbb{E}H(\nu N,\lf \eta N\rf,\tau N).
$$
Thus, $h$ describes the deterministic limit shape. It can be described explicitly as follows. Let $G(z)=G(\nu,\eta,\tau;z)$ be the function
\begin{equation}\label{Stdl}
G(\nu,\eta,\tau;z) = \tau\frac{z+z^{-1}}{2} + \eta\log \left( \frac{z+z^{-1}}{2}-1\right) - \nu\log z.
\end{equation}
There is an explicit (in the sense that it can be written in terms of algebraic functions) connected domain $\mathcal{D}$ consisting of all triples $(\nu,\eta,\tau)$ such that $G(\nu,\eta,\tau;z)$ has a unique critical point in the region $\mathbb{H}-\mathbb{D}=\{z:\Im z>0 \text{ and } \vert z\vert>1\}$. This induces a map $\Omega:\mathcal{D}\rightarrow \mathbb{H}-\mathbb{D}$ by sending $(\nu,\eta,\tau)$ to the critical point of $G(z)$. Then 
$$
h(\nu,\eta,\tau)=\Im\left(\frac{S(\Omega(\nu,\eta,\tau))}{2\pi}\right).
$$
Outside of $\mathcal{D}$, the limit shape is trivial -- that is, if $\nu $ is too large, then there are no particles to the right of $(\nu N,\eta N)$ at time $\tau N$, so the height function is zero. If $\eta$ is too small, then all the particles are to the right of $(\nu N,\eta N)$ at time $tN$, so the height function is $\eta/2$. In the literaure, $\mathcal{D}$ is called the \textit{liquid region} and the triples $(\nu,\eta,\tau)$ outside of $\mathcal{D}$ is called the \textit{frozen region}.

\smallskip

\textbf{Gaussian free field fluctuations}
In order to describe the fluctuations, let us review the Gaussian free field. A comprehensive survey can be found in \cite{kn:S}. The Gaussian free field is a Gaussian probability measure on a suitable class of distributions on a domain $D\subset \R^d$. More precisely, given compactly supported smooth test functions $\{\phi_m\}_{m=1}^{\infty}$, the random variables $\{\text{GFF}(\phi_m)\}_{m=1}^{\infty}$ are mean zero Gaussians with covariance
\begin{equation}\label{GFFdefn}
\mathbb{E}[\text{GFF}(\phi_{m_1})\text{GFF}(\phi_{m_2})] = \int_{D\times D} \phi_{m_1}(z_1)\phi_{m_2}(z_2)\mathcal{G}(z_1,z_2)dz_1dz_2,
\end{equation}
where $\mathcal{G}(z_1,z_2)$ is the Green function for the Laplacian on $D$ with Dirichlet boundary conditions. 

Formally, one could attempt to set $\phi_{m}=\delta_{z_m}$ for $ z_m\in D$ in order to define the Gaussian free field at a point. However \eqref{GFFdefn} would imply that $\text{GFF}(z)$ has variance $\mathcal{G}(z,z)$, which is undefined for $d\geq 2$. However, for pairwise distinct points $z_1,\ldots,z_k$ one expects from Wick's theorem
$$
\displaystyle
\mathbb{E}[\text{GFF}(z_1)\ldots\text{GFF}(z_k)] =
\begin{cases}
\sum_{\sigma} \prod_{i=1}^{k/2} \mathcal{G}(z_{\sigma(2i-1)},z_{\sigma(2i)}), \quad &\text{k even }\\
0, \quad &\text{ k odd}
\end{cases}
$$
where the sum is over all fixed point free involutions $\sigma$ on $\{1,\ldots,k\}$. This can be made into a rigorous statement:
$$
\mathbb{E}[\text{GFF}(\phi_1)\ldots\text{GFF}(\phi_k)] = \int_{D^n}\mathbb{E}[\text{GFF}(z_1)\ldots\text{GFF}(z_k)] \prod_{i=1}^k \phi_i(z_i)dz_i.
$$
Furthermore, these moments uniquely determine the Gaussian free field.

\begin{theorem}\label{OrthTheorem}
Let $\varkappa_j=(\nu_j,\eta_j,\tau)\in \mathcal{D}$ for $1\leq j\leq k$. Define
$$
H_N(\nu,\eta,\tau) := \frac{1}{N}\left(H(\nu N,\lf \eta N\rf, \tau N) - \mathbb{E}H(\nu N,\lf \eta N\rf, \tau N) \right)
$$
and let $\Omega_j=\Omega(\varkappa_j).$ Then
$$
\lim_{N\rightarrow \infty} \mathbb{E}(H_N(\varkappa_1)\cdots H_N(\varkappa_k)) = 
\begin{cases}
\sum_{\sigma} \prod_{i=1}^{k/2} \mathcal{G}(\Omega_{\sigma(2i-1)},\Omega_{\sigma(2i)}), \quad &\text{ k even }\\
0, \quad &\text{ k odd}
\end{cases}
$$ 
where the sum is over all fixed point free involutions $\sigma$ on $\{1,\ldots,k\}$ and $\mathcal{G}$ is the Green's function for the Laplacian on $\mathbb{H}-\mathbb{D}$ with Dirichlet boundary conditions:
\begin{equation}\label{Dirichlet}
\mathcal{G}(z,w) = \frac{1}{2\pi}\log\left(\frac{z+z^{-1}-\bar{w}-\bar{w}^{-1}}{z+z^{-1}-w-w^{-1}} \right).
\end{equation}
\end{theorem}

\smallskip

\textbf{Idea of proof and generalization}
The proof uses a very specific property of the interacting particle system, namely that it is a \textit{determinantal point process}. There are several previous examples of determinantal point processes having Gaussian free field fluctuations \cite{kn:BF,kn:CJY,kn:D,kn:P}. (See also \cite{kn:K}The essential idea in these proofs is similar.. One takes an explicit formula for the correlation kernel $K(x,y)$, and then asymptotic analysis on $K(x,x)$ provides information about the limit shape while asymptotics of $K(x,y),x\neq y$ provides information about the fluctuations. In \cite{kn:BK1}, an explicit formula for the correlation kernel was proved, enabling steepest descent analysis. 

It is thus natural to ask: given a determinantal point process with an explicit correlation kernel, is there a general statement that the fluctuations of the height function are governed by the Gaussian free field? The answer is yes.

\begin{theorem}
Suppose we are given a particle system on $\Z\times\Z_{\geq 1}$ which is a determinantal point process with correlation kernel
$$
K(x_1,n_1,x_2,n_2,t)\approx
\left(\frac{1}{2\pi i}\right)^2\int_{\Gamma_1}\int_{\Gamma_2}\frac{\exp(NG(\nu_1,\eta_1,\tau,u))}{\exp(NG(\nu_2,\eta_2,\tau,w))}f(u,w)dwdu,
$$
where $\Gamma_1,\Gamma_2$ are steepest descent paths. We make certain technical assumptions about $K$ (see Definition \label{Normal} below).

Let $\mathcal{D}\subset\R^3$ be the liquid region and let $\Omega:\mathcal{D}\rightarrow\C$ send $(\nu,\eta,\tau)$ to the critical point of $G(\nu,\eta,\tau,z)$. 
If $H_N$ denotes the scaled and centered random height function of the particle system, then for $\varkappa_1,\ldots,\varkappa_k\in \mathcal{D}$ with $\Omega_j=\Omega(\varkappa_j)$
\[
\E[H_N(\varkappa_1)\ldots H_N(\varkappa_k)]\rightarrow
\begin{cases}
\displaystyle\sum_{\sigma}\prod_{j=1}^{k/2} \mathcal{G}(\Omega_{\sigma(j)},\Omega_{\sigma(j+1)}),\ \ k \ \text{even}\\
0,\ \ k\ \text{odd},
\end{cases}
\]
where 
\begin{equation}\label{GeneralG}
\mathcal{G}(z,w)=\frac{1}{2\pi}\int_{\bar{z}}^z\int_{\bar{w}}^w\frac{f(u,v)f(v,u)}{G_{\nu}'(u)G_{\nu}'(v)}dudv,
\end{equation}
with $G_{\nu}'$ denoting $(\partial^2/\partial \nu\partial z) G.$
\end{theorem}
The rigorous details are in Section \ref{GenRes}. In particular, the formula for $\mathcal{G}$ in \eqref{Dirichlet} follows from \eqref{GeneralG} with $S$ as in \eqref{Stdl} and 
$$
f(u,v) = \frac{1}{v} \frac{1-u^{-2}}{v+v^{-1}-u-u^{-1}}.
$$

\smallskip

\textbf{Outline of paper} 
In section \ref{Statement of the Main Theorem}, we state precisely the assumptions on the determinantal point process, as well as explain why these assumptions are natural. In sections \ref{AlgebraGeneral} and \ref{AnalysisGeneral}, we prove Theorem \ref{GeneralG}. In section \ref{Specific Results}, we show that Theorem \ref{OrthTheorem} follows once we prove that the interacting particle system with a reflecting wall satisfies the necessary technical assumptions. In section \ref{SpecificAlgebra} and \ref{SpecificAnalysis}, we show that the necessary technical assumptions indeed hold. Section 4 collects the asymptotic analysis needed throughout the proofs. 

\textbf{Acknowledgments} The author would like to thank Alexei Borodin and Maurice Duits for insightful comments. The author was supported by the National Science Foundation's Graduate Research Fellowship Program.

\section{General Results}\label{GenRes}
\subsection{Statement of the Main Theorem}\label{Statement of the Main Theorem}
Suppose we have a family of point processes on $\X=\Z\times\Z_{\geq 1}$ which runs over time $t\in [0,\infty)$. (Note that these are different co-ordinates from the introduction). In other words, at any time $t$, the system selects a random subset $X\subset\X$. If $(x,n)\in X$, then we say that there is a particle at $(x,n)$. For any $k\geq 1$ and $t\geq 0$, let $\rho_k^t:\X^k\rightarrow [0,1]$ be defined by
\begin{multline*}
\rho_k^t(x_1,n_1,\ldots,x_k,n_k) \\
= \mathbb{P}(\text{There is a particle at}\ (x_j,n_j)\ \text{at time}\ t\ \text{for each}\ j=1,\ldots,k).
\end{multline*}
Assume that there is a map $K$ on $\X\times\X\times [0,\infty)$ such that
\begin{equation}\label{DeterminantalCorrelationFunction}
\rho_k^t(x_1,n_1,\ldots,x_k,n_k) = \det[K(x_i,n_i,x_j,n_j,t)]_{1\leq i,j\leq k}.
\end{equation}
The maps $\rho_k$ and $K$ are called the \textit{kth correlation function} and the \textit{correlation kernel}, respectively. 

A function $c$ on $\X\times\X$ is called a \textit{conjugating factor} if there exists another function $\mathcal{C}$ on $\X$ such that
\[
c(x,n,x',n')=\frac{\mathcal{C}(x,n)}{\mathcal{C}(x',n')}.
\]
Note that if $c$ is a conjugating factor, then
\begin{equation}\label{Conjugating}
\det[K(x_i,n_i,x_j,n_j,t)]_{1\leq i,j\leq k} = \det[c(x_i,n_i,x_j,n_j)K(x_i,n_i,x_j,n_j,t)]_{1\leq i,j\leq k}.
\end{equation}
Two kernels $K$ and $\tilde{K}$ are called \textit{conjugate} if $\tilde{K}=cK$ for some conjugating factor $c$.

If a correlation kernel exists, the point process is called \textit{determinantal}. On a discrete space, a point process is determined uniquely by its correlation functions (see e.g. \cite{kn:L}).  Therefore, if we are given two determinantal point process on a discrete space with conjugate kernels, they must have the same law.

The set $\Z\times\{n\}$ is called the \textit{nth level}. Given a subset $X\subset\X$, let $m_n$ be the cardinality of the set $X\cap (\Z \times \{n\})$. Assume that the numbers $m_n$ take constant finite values which are independent of the time parameter $t$. In words, this means that the number of particles on the $n$th level is always $m_n$. Further assume that $m_n\leq m_{n+1}\leq m_n+1$ for all $n$. Let $x_1^{(n)} > x_2^{(n)} > \ldots > x_{m_n}^{(n)}$ denote the elements of $X\cap (\Z \times \{n\})$. A subset $X$ is called \textit{interlacing} if 
\[
x_{k+1}^{(n+1)} < x_k^{(n)} \leq x_k^{(n+1)}, \ \ \text{when} \ \ m_{n+1}=m_{n},
\]
\[
x_{k+1}^{(n+1)} \leq x_k^{(n)} < x_k^{(n+1)}, \ \ \text{when} \ \ m_{n+1}=m_{n}+1.
\]
Assume that at any time $t$, the system almost surely selects an interlacing subset.  Let $\delta_n$ equal $m_{n+1}-m_n$.

Define the random \textit{height function} by 
\[
h: \X \times \mathbb{R}_{\geq 0} \rightarrow \mathbb{Z}_{\geq 0},
\]
\[
h(x,n,t) = \vert \{(s,n)\in X: s>x\} \vert.
\]
In words, $h$ counts the number of particles to the right of $(x,n)$ at time $t$. 

We wish to study the large-time asymptotics of this particle system. Let $x=[N\nu], n=[N\eta], t=N\tau$, where $N$ is a large parameter. Define $\mathcal{D}\subset\mathbb{R}\times\R_+\times\R_+$ to be 
\[
\mathcal{D}:=\{(\nu,\eta,\tau): \lim_{N\rightarrow\infty} \rho_1^t(x,n)>0\}.
\]
Let $H_N$ be defined by
\[
H_N: \mathbb{R}\times\R_+\times\R_+ \rightarrow \R,
\]
\[
H_N(\nu,\eta,\tau) := h(x,n,t) -\mathbb{E}h(x,n,t).
\]
In words, $H_N$ is the fluctuation of the height function around its expectation.

Before stating the theorem, we need to state some more assumptions on the kernel.

Suppose the kernel $K$ is conjugate to a kernel $\tilde{K}$ such that $\tilde{K}$ satisfies the following property: There is a number $L$ such that whenever $x,x'\geq L$,
\begin{multline}
\tilde{K}(x,n,x',n',t)+\tilde{K}(x,n,x'-1+\dn,n'+1,t)+\tilde{K}(x,n,x'+\dn,n'+1,t)\\
=
\begin{cases}\label{C1}
1,\ \ (x,n)=(x',n')\\
0,\ \ \text{otherwise},
\end{cases}
\end{multline}
\begin{multline}
\tilde{K}(x,n,x',n',t)+\tilde{K}(x+1-\dn,n-1,x',n',t)+\tilde{K}(x-\dn,n-1,x',n',t)\\
=
\begin{cases}\label{C2}
1,\ \ (x,n)=(x',n')\\
0,\ \ \text{otherwise}.
\end{cases}
\end{multline}
Further suppose that for $x',x''>L$, 
\begin{align}
\tilde{K}(x,n,x',n',t)\tilde{K}(x'',n'',x-1+\dn,n+1,t) \rightarrow 0\ &\text{as}\ x\rightarrow\infty\label{C3}\\
\tilde{K}(x,n,x',n',t)\tilde{K}(x'',n'',x+\dn,n+1,t) \rightarrow 0\ &\text{as}\ x\rightarrow\infty\label{C4}\\
\tilde{K}(x,n,x-1+\dn,n+1,t) \rightarrow 0\ &\text{as}\ x\rightarrow\infty\label{C6}\\
\tilde{K}(x,n,x+\delta_n,n+1,t) \rightarrow 1\ &\text{as}\ x\rightarrow\infty.\label{C5}
\end{align}

Suppose $G(\nu,\eta,\tau,z)$ is a complex-valued function on $\mathbb{R}\times\R_+\times\R_+\times\mathbb{C}$. To save space, we will sometimes write $G(z)$. Expressions such as $G',G_{\nu},G'_{\nu}$ will be shorthand for $\partial G/\partial z$, $\partial G/\partial \nu$ and $\partial^2 G/\partial z\partial\nu$, respectively. Assume $G(\overline{z})=\overline{G(z)}$.
Also suppose there exists a differentiable map $\Omega$ from $\mathcal{D}$ to the upper half-plane $\mathbb{H}=\{z\in\mathbb{C}:\Im(z)>0\}$ such that $\Omega$ is a critical point of $G$. In other words, 

\begin{equation}
G'(\nu,\eta,\tau,\Omega(\nu,\eta,\tau))=0\ \ \text{for all}\ \ (\nu,\eta,\tau)\in\mathcal{D}.
\label{DefinitionOfOmega}
\end{equation}
Note that $\Om$ need not be onto. For any $(\eta,\tau)$, if the set $\{\nu\in\mathbb{R}:(\nu,\eta,\tau)\in\mathcal{D}\}$ is nonempty, let $q_2(\eta,\tau)$ denote its supremum.

\begin{definition}\label{Normal}
With the notation above, a determinantal point process on $\Z\times\Z_{\geq 1}$ is \textit{normal} if all of the following hold:

\begin{itemize}
\item
For all $\eta,\tau>0$, the limit $\Om(q_2(\eta,\tau)-0,\eta,\tau)$ exists and is a positive real number.

\item 
For all $\eta,\tau>0$, as $\nu$ approaches $q_2(\eta,\tau)$ from the left, $G''(\nu,\eta,\tau,\Om(\nu,\eta,\tau))=\bigO((q_2(\eta,\tau)-\nu)^{1/2})$.

\item
$K$ is conjugate to some $\tilde{K}$ such that \eqref{C1}-\eqref{C5} hold for some integer $L$.


\item
Set $t=N\tau$, $x_j=[N\nu_j]$ and $n_j=[N\eta_j]$ for $j=1,2$, where $(\nu_j,\eta_j,\tau)\in\mathcal{D}$. Let $\Om_j$ denote $\Om(\nu_j,\eta_j,\tau)$ and let $G_j(z)$ denote $G(\nu_j,\eta_j,\tau,z)$. If $\Om_1\neq\Om_2$ and $k_1$,$k_2$ are finite integers, then as $N\rightarrow\infty,$
\begin{multline}
\tilde{K}(x_1,n_1+k_1,x_2,n_2+k_2,t)=\\
\left(\frac{1}{2\pi i}\right)^2\int_{\Gamma_1}\int_{\Gamma_2}\frac{\exp(NG(\nu_1,\eta_1,\tau,z))}{\exp(NG(\nu_2,\eta_2,\tau,w))}f_{k_1k_2}(u,w)dwdu+\bigO(e^{N\kappa}),
\label{AssumptiononK}
\end{multline}
where $\Gamma_1$ and $\Gamma_2$ are steepest descent paths, $\kappa< \Re(G_1(\Om_1) -G_2(\Om_2))$, and $f_{k_mk_n}$ are complex-valued meromorphic functions satisfying the identity
\begin{multline*}
f_{k_1k_2}(z_1,z_2)f_{k_2k_3}(z_2,z_3)\ldots f_{k_{r-1}k_r}(z_{r-1},z_r)f_{k_rk_1}(z_r,z_1)\\
=f(z_1,z_2)f(z_2,z_3)\ldots f(z_{r-1},z_r)f(z_r,z_1).
\end{multline*}
Here, we have written $f$ for $f_{00}$.

\item
For any $l\geq 3$, the following indefinite integral satisfies
\begin{equation}\label{Cycle}
\displaystyle\int\cdots\int\sum_{\sigma}\displaystyle\prod_{i=1}^l \frac{f(z_{\sigma(i)},z_{\sigma(i+1)})}{G_{\nu}'(z_{\sigma(i)})}dz_i \equiv 0,
\end{equation}
where the sum is taken over all $l$-cycles in $S_l$ and the indices are taken cyclically.

\item
For any finite interval $[a,b]$, $G\in C^2[a,b]$ and the Lesbesgue measure of the set $\{x\in [a,b]: I'(x)\in 2\pi\Z+[-\delta,\delta]\}$ is $\bigO(\delta^a)$ for some positive $a$. 
\end{itemize}
\end{definition}

The following remarks will help explain the definition.
\begin{remark}\label{NormalExplanation}
(1) The assumption that $\Om(q_2(\eta,\tau)-0,\eta,\tau)>0$ occurs naturally. One often finds that for $k_1\neq k_2\in\Z$,
\[
\displaystyle\lim_{N\rightarrow\infty} K([N\nu]+k_1,[N\eta],[N\nu]+k_2,[N\eta],N\tau)=\frac{1}{2\pi i}\int_{\bar{\Om}}^{\Om} \frac{dz}{z^{k_1-k_2+1}}=\frac{\Im(\Om^{k_2-k_1})}{\pi(k_2-k_1)},
\]
where the contour crosses the positive real line. By setting $\Om=e^{i\varphi}$, we see that the right hand side reduces to the ubiquitous sine kernel. When $k_1=k_2=0$, we see that
\[
\lim_{N\rightarrow\infty} \rho_1^{N\tau}([N\nu],[N\eta])=\frac{1}{2\pi}(\log\Om-\log\bar{\Om})=\frac{\arg\Om(\nu,\eta,\tau)}{\pi}.
\]
Since the left hand side equals zero, we expect $\arg\Om(q_2(\eta,\tau)-0,\eta,\tau)=0$.

(2) Since $G(\overline{z})=\overline{G(z)}$, this means that $G$ has a critical point at both $\Om$ and $\bar{\Om}$. As $\Om$ approaches the real line, the two critical points coalesce into a triple zero, so $G''(t,\eta,\tau)$ converges to $0$ as $t$ approaches $q_2(\eta,\tau)$. We need a control for how quickly this convergence to $0$ occurs, in order to order to control the behavior near the boundary of $\mathcal{D}$. More specifically, it controls the bound in Proposition \ref{EdgeAsymptotics}.

There is a heuristic understanding for why (2) should hold. The function $G$ has two critical points which coalesce into a triple zero. The simplest example of such a function is $G(t,x)=x^3/3-tx$ as $t$ approaches $0$. In this case, the solution to $G'(t,x)=0$ is $\Om(t,x)=t^{1/2}$. Then $G''(t,\Om(t,x))=2t^{1/2}$.  


(3) Assumptions \eqref{C1}--\eqref{C5} will be elucidated when we interpret the particles as lozenges. In particular, see remark \ref{LozengeRemark}.

(4) It is common for the kernel to be expressed in this form; previous examples are \cite{kn:BK1} and \cite{kn:BF}. If the kernel has a different expression with the same asymptotics as in Propositions \ref{Prop6.9} and \ref{EdgeAsymptotics}, the results still hold.

(5) In particular, \eqref{Cycle} holds if there always exist $u$-substitutions and an expression $Y$ such that  
\[
\int\cdots\int \displaystyle\prod_{i=1}^l \frac{f(z_i,z_{i+1})}{G_{\nu_1}'(z_i)} dz_i =\int\cdots\int \displaystyle\prod_{i=1}^l \frac{1}{Y(u_i)-Y(u_{i+1})}du_i,
\]
where $z_{l+1}=z_1$ and $u_{l+1}=u_1$. This is because of Lemma 7.3 in \cite{kn:K}, which refers back to \cite{kn:B}, which says that 
\[
\displaystyle\sum_{\sigma}  \prod_{i=1}^l \frac{1}{Y(u_{\sigma(i)})-Y(u_{\sigma(i+1)})}=0.
\]

(6) This is a technical lemma which allows Lemma \ref{ZeroLemma} to be applied.
\end{remark}

We can now state the main theorem.
\begin{theorem}\label{MainTheorem}
Suppose we are given a normal determinantal point process. For $1\leq j\leq k$, let $\varkappa_j=(\nu_j,\eta_j,\tau)$ be distinct points in $\mathcal{D}$ , and let $\Om_j=\Om(\nu_j,\eta_j,\tau)$. Define the function $\mathcal{G}$ on the upper half-plane to be
\[
\mathcal{G}(z,w)=\left(\frac{1}{2\pi}\right)^2\int_{\bar{z}}^z\int_{\bar{w}}^w\frac{f(z_1,z_2)f(z_2,z_1)}{G_{\nu}'(z_1)G_{\nu}'(z_2)}dz_2dz_1
\]
Then
\[
\displaystyle\lim_{N\rightarrow\infty}\mathbb{E}(H_N(\varkappa_1)\cdots H_N(\varkappa_k))=
\begin{cases}
\displaystyle\sum_{\sigma\in\mathcal{F}_k}\prod_{j=1}^{k/2}\mathcal{G}(\Om_{\sigma(2j-1)},\Om_{\sigma(2j)}),\ \ &\text{k is even}\\
0,\ \ & \text{k is odd},
\end{cases}
\]
where $\mathcal{F}_k$ is the set of all involutions in $S_k$ without fixed points.
\end{theorem}

\begin{remark} We note that these are the moments of a linear family of Gaussian random variables: see Appendix A. Using the results of \cite{kn:So}, it should be possible to show that $H_N(\varkappa)/\sqrt{\text{Var} H_N(\varkappa)}$ converges to a Gaussian, but this was not pursued. 
\end{remark}

\subsection{Algebraic steps in proof of Theorem \ref{MainTheorem}}\label{AlgebraGeneral}

The most natural way to view $\X$ is as a square lattice. However, it turns out that a hexagonal lattice is more useful. To obtain the hexagonal lattice, take the $n$th level and shift it to the right by $(n+1)/2-m_n$. See Figure \ref{Particles}. 

Figure \ref{Particles} also shows that the particle system can be interpreted as lozenges. Each lozenge is a pair of adjacent equilateral triangles. See Figure \ref{Lozenges}. 

By setting the location of each triangle to be the midpoint of its horizontal side, each lozenge can be viewed as a pair $(x,n,x',n')$, where the black triangle is located at $(x,n)$ and the white triangle is located at $(x',n')$. For example, in Figure \ref{Particles} there are lozenges $(1,3,1,3), (2,3,2,4),$ and $(0,3,1,4)$. The three types of lozenges can be described as follows. For lozenges of type I, $(x',n')=(x,n)$. For lozenges of type II, $(x',n')=(x-1+\delta_n,n+1)$. For lozenges of type III, $(x',n')=(x+\delta_n,n+1)$. Note that a lozenge of type I is just a particle. 

We say that $(x,n,x',n')\in\mathfrak{X}\times\mathfrak{X}$ is \textit{viable} if $(x',n')=(x,n), (x-1+\delta_n,n+1)$, or $(x+\delta_n,n+1)$. A sequence $(x_1,n_1,x_1',n_1'),\ldots,(x_k,n_k,x_k',n_k')$ of viable elements is \textit{non-overlapping} if $(x_1,n_1),\ldots,(x_k,n_k)$ are all distinct from each other and $(x_1',n_1'),\ldots,(x_k',n_k')$ are also all distinct from each other. We do, however, allow the possibility of $(x_i,n_i)=(x_j',n_j')$.

\begin{figure}
\caption{In this example, the integers $m_n$ equal $1,1,1,2,3,3,4\ldots$. The black line on the left represents the points where $x=0$. Examples of $x_{k}^{(n)}$ are $x_1^{(3)}=1,x_1^{(4)}=3,x_2^{(7)}=4$.}
\begin{center}
\includegraphics[height=2in]{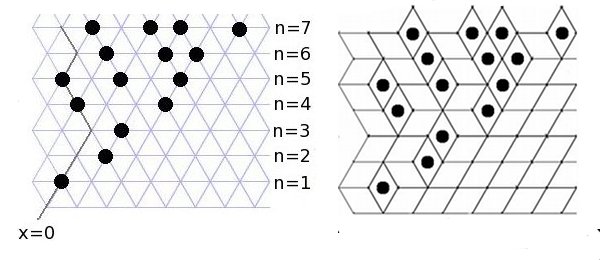}
\end{center}
\label{Particles}
\end{figure}

\begin{figure}
\caption{Lozenges of types I,II, and III, respectively. Note that lozenges of type I occur exactly at the same places as particles.}
\centering
\includegraphics[height=0.4in]{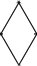}%
\hspace{1in}%
\includegraphics[height=0.2in]{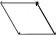}%
\hspace{1in}%
\vspace{0.2in}%
\includegraphics[height=0.2in]{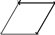}%

\includegraphics[height=0.4in]{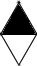}%
\hspace{1in}%
\includegraphics[height=0.2in]{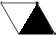}%
\hspace{1in}%
\includegraphics[height=0.2in]{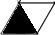}%
\label{Lozenges}
\end{figure}

The statement and proof of the next proposition are similar to Theorem 5.1 of \text{kn:BF}.

\begin{proposition}\label{LozengeKernel}
Suppose the kernel $K$ is conjugate to some $\tilde{K}$ such that \eqref{C1}-\eqref{C5} hold for some $L$. If  $t\geq 0, x_1,x_1',\ldots,x_k,x_k'>L,$ and $(x_1,n_1,x_1',n_1'),\ldots,$ $(x_k,n_k,x_k',n_k')$ is a sequence of non-overlapping viable elements of $\X\times\X$, then 
\begin{multline}\label{PEqualsD}
\mathbb{P}(\text{There is a lozenge}\ (x_j,n_j,x_j',n_j')\ \text{at time}\ t\ \text{for each}\ j=1,\ldots,k)\\
=\det[\tilde{K}(x_i,n_i,x_j',n_j',t)]_{1\leq i,j\leq k}.
\end{multline}
\end{proposition}
\begin{remark}\label{LozengeRemark}
The equations \eqref{C1}--\eqref{C5} can now be intuitively understood. Equation \eqref{C1} says that each black triangle is located in exactly one of the three lozenges around it, and equation \eqref{C2} makes an identical statement for white triangles. Equations \eqref{C3} and \eqref{C6} say that lozenges of type II almost surely do not occur far to the right of the particles, with \eqref{C3} controlling the off-diagonal entries in the determinant and \eqref{C6} controllling the diagonal entries. Similarly, equations \eqref{C4} and \eqref{C5} says that lozenges of type III almost surely do occur far to the right of the particles. This intuition will be exploited in the proof of Thereom \ref{LozengeKernel}.
\end{remark}
\begin{proof}
We proceed by induction on the number of lozenges that are not of type I. When this number is zero, the statement reduces to \eqref{DeterminantalCorrelationFunction} and \eqref{Conjugating}.

For any set $S=\{(x_1,n_1,x_1',n_1'),\ldots,(x_k,n_k,x_k',n_k')\}$ of non-overlapping, viable elements, let $P(S)$ and $D(S)$ denote the left and right hand sides of \eqref{PEqualsD}, respectively. First, as a preliminary statement, it is not hard to prove that if $(x_{k+1},n_{k+1})\neq (x_r,n_r)$ for $1\leq r\leq k$, then 
\begin{multline}\label{DEquation1}
D(S\cup\{(x_{k+1},n_{k+1},x_{k+1},n_{k+1})\})+D(S\cup\{(x_{k+1},n_{k+1},x_{k+1}-1+\dn,n_{k+1}+1)\})\\
+D(S\cup\{(x_{k+1},n_{k+1},x_{k+1}+\dn,n_{k+1}+1)\})=D(S).
\end{multline}
One simply expands the determinant in the left--hand--side as a sum over permutations $\sigma\in S_{k+1}$. One then uses \eqref{C1} to show that the sum over the $\sigma$ fixing $k+1$ equals $D(S)$, while the sum over the sigma not fixing $k+1$ equals $0$. Note that if $D$ is replaced by $P$ in \eqref{DEquation1}, the statement is immediate, since the black triangle at $(x_{k+1},n_{k+1})$ must be contained in exactly one lozenge. 


In a similar manner, if $(x_{k+1}',n_{k+1}')\neq (x_r',n_r')$ for $1\leq r\leq k$, then \eqref{C2} implies that 
\begin{multline}\label{DEquation2}
D(S\cup\{(x_{k+1},n_{k+1},x_{k+1},n_{k+1})\})+D(S\cup\{(x_{k+1}+1-\dn,n_{k+1}-1,x_{k+1},n_{k+1})\})\\
+D(S\cup\{(x_{k+1}-\dn,n_{k+1}-1,x_{k+1},n_{k+1})\})=D(S).
\end{multline}
Again, the statement holds if $D$ is replaced by $P$.

In order to prove the induction step, it suffices to prove that $D$ and $P$ still agree if we add a lozenge of type $II$ or type $III$ to $S$. Let us do type $II$, as type $III$ is similar.
Suppose that $(x,n,x-1+\dn,n+1)$ is viable and that $S\cup\{(x,n,x-1+\dn,n+1)\}$ is non-overlapping. Then equation \eqref{DEquation1} is equivalent to
\begin{multline}\label{BEquivalent}
D(S\cup\{(x,n,x-1+\dn,n+1)\})\\
= D(S)-D(S\cup\{(x,n,x,n)\})-D(S\cup\{(x,n,x+\dn,n+1)\}),
\end{multline}
and the same holds for $P$ instead of $D$.
By the induction hypothesis, 
\begin{align*}
D(S) &= P(S),\\
D(S\cup\{(x,n,x,n)\}) &= P(S\cup\{(x,n,x,n)\})\\
D(S\cup\{(x+\dn,n+1,x+\dn,n+1)\})&=P(S\cup\{(x+\dn,n+1,x+\dn,n+1)\}).
\end{align*}
Thus, \eqref{BEquivalent} implies 
\begin{multline*}
D(S\cup\{(x,n,x-1+\dn,n+1)\})-P(S\cup\{(x,n,x-1+\dn,n+1)\}))\\
=-D(S\cup\{(x,n,x+\dn,n+1)\})+P(S\cup\{(x,n,x+\dn,n+1)\}).
\end{multline*}

Assume for now that $(x_r',n_r')\neq (x+\delta_n,n+1)$ for $1\leq r\leq $k. Then we cam apply equation \eqref{DEquation2}, which implies that
\begin{multline}\label{WEquivalent}
D(S\cup\{(x,n,x-1+\dn,n+1)\})=D(S)\\
-D(S\cup\{(x+\dn,n+1,x+\dn,n+1)\})-D(S\cup\{(x+1,n,x+\dn,n+1)\}),
\end{multline}
and the same statement holds for $P$. Thus, 
\begin{multline*}
-D(S\cup\{(x,n,x+\dn,n+1)\})+P(S\cup\{(x,n,x+\dn,n+1)\})\\
=D(S\cup\{(x+1,n,x+\dn,n+1)\})-P(S\cup\{(x+1,n,x+\dn,n+1)\}).
\end{multline*}

If $S\cup\{(x+1,n,x+\delta_n,n+1)\}$ is non-overlappinng, then \eqref{BEquivalent} is again applicable. We repeatedly apply \eqref{BEquivalent} and \eqref{WEquivalent} as often as possible. First, suppose that this can be done indefinitely. Then 
\begin{multline*}
\vert D(S\cup\{(x,n,x-1+\dn,n+1)\})-P(S\cup\{(x,n,x-1+\dn,n+1)\}) \vert\\
=\displaystyle\lim_{M\rightarrow\infty} \vert D(S\cup\{(x+M,n,x-1+\dn+M,n+1)\})-P(S\cup\{(x+M,n,x-1+\dn+M,n+1)\}) \vert.
\end{multline*}
Since lozenges of type II almost surely do not appear when we look far to the right of the particles, 
\[
\displaystyle\lim_{M\rightarrow\infty} P(S\cup\{(x+M,n,x-1+\dn+M,n+1)\})=0.
\]
By expanding the determinant into a sum over $S_{k+1}$, \eqref{C3} and \eqref{C6} imply that 
\[
\displaystyle\lim_{M\rightarrow\infty} D(S\cup\{(x+M,n,t,x-1+\dn+M,n+1)\})=0.
\]

Now suppose that \eqref{BEquivalent} and \eqref{WEquivalent} can only be applied finitely many times. This means that $D(S\cup\{(x,n,x-1+\dn,n+1)\})-P(S\cup\{(x,n,x-1+\dn,n+1)\})$ equals either
\[
D(S\cup\{(x+M,n,x+M+\dn,n+1)\})-P(S\cup\{(x+M,n,x+M+\dn,n+1)\})
\]
or
\[
D(S\cup\{(x+M+1,n,x+M+\dn,n+1)\})-P(S\cup\{(x+M+1,n,x+M+\dn,n+1)\})
\]


In the first case, $S\cup\{(x+M+1,n,x+M+\dn,n+1)\}$ is non non-overlapping. This implies $D(S\cup\{(x+M+1,n,x+M+\dn,n+1)\})=0$ (because two of the rows are idential) and $P(S\cup\{(x+M+1,n,x+M+\dn,n+1)\})=0$ (because a triangle cannot be in two different lozenges at the same time). Thus, $D$ and $P$ agree. A similar argument holds in the second case. Thus, $D$ and $P$ agree whenever a lozenge of type II is added to $S$. 

An identical argument holds for type III lozenges, except that we use \eqref{C4} and \eqref{C5} instead of \eqref{C3} and \eqref{C6}.
\end{proof}

We have been describing a lozenge as a pair $(x,n,x',n')$. It can also be described as $(x',n',\lambda)$, where $(x',n')$ is the location of the white triangle and $\lambda\in\{I,II,III\}$ is the type of the loznege. Thus the proposition can be restated as the following statement.

\begin{corollary}\label{Alternative}
For any non--overlapping $(x'_1,n'_1,\lambda_1),\ldots,(x'_k,n'_k,\lambda_k)$,
\begin{multline*}
\mathbb{P}(\text{There is a lozenge}\ (x_j',n_j',\lambda_j)\ \text{at time t for each}\ j=1,\ldots,k)\\
=\det[K_{\lambda}(x_i',n_i',\lambda_i,x_j',n_j',t)]_{1\leq i,j\leq k},
\end{multline*}
where 
\begin{multline*}
K_{\lambda}(x,n,\lambda,x',n',t)=
\begin{cases}
\tilde{K}(x,n,x',n',t),\ \text{when}\ \lambda=I\\
\tilde{K}(x-\delta_{n-1},n-1,x',n',t),\ \text{when}\ \lambda=II\\
\tilde{K}(x-\delta_{n-1}-1,n-1,x',n',t),\ \text{when}\ \lambda=III
\end{cases}
\end{multline*}
\end{corollary}
\begin{proof}
This is a result of the correspondences
\begin{align*}
(x',n',I)\ &\text{iff}\ (x',n',x',n'),\\
(x',n',II)\ &\text{iff}\ (x'-\delta_{n'-1},n'-1,x',n'),\\
(x',n',III)\ &\text{iff}\ (x'-\delta_{n'-1}-1,n'-1,x',n').
\end{align*}
\end{proof}

There are two different formulas for the height function. One formula is
\begin{equation}\label{HeightFormula1}
h(x,n)=\displaystyle\sum_{s>x}\mathbf{1}(\text{lozenge of type I at}\ (s,n)).
\end{equation}
It is possible to only use \eqref{HeightFormula1} to complete the proof. However,  when there are multiple points on one level, i.e. not all $\eta_1,\ldots,\eta_k$ are distinct, the computation becomes much more complicated. This is because lozenges of type I will appear in multiple sums of the form \eqref{HeightFormula1}. We can avoid this difficulty by introducing another formula for the height function:
\begin{equation}\label{HeightFormula2}
h(x,n)=h(x+\delta_n+\delta_{n+1}+\ldots+\delta_{n'-1},n')+H_{n,n'}(x),
\end{equation}
where, for $n<n'$,
\begin{equation}\label{HeightFormula3}
H_{n,n'}(x)=-\displaystyle\sum_{p=n+1}^{n'} \mathbf{1}(\text{lozenge of type II at}\ (x+\delta_n+\delta_{n+1}\ldots+\delta_{p-1},p)).
\end{equation}
Therefore, the expression 
\begin{equation}\label{Expectation}
\mathbb{E}\left(\displaystyle\prod_{j=1}^k [h(x_j,n_j)-\mathbb{E}(h(x_j,n_j))]\right)
\end{equation}
can be expressed as a sum of terms of the form
\begin{equation}\label{EExpression}
\mathbb{E}\left(\displaystyle\prod_{j=1}^{k'} [h(x_j,n_j)-\mathbb{E}(h(x_j,n_j))]\prod_{l=k'+1}^k [H_{n_l,n_l'}(x_l)-\mathbb{E}(H_{n_l,n_l'}(x_l))]\right).
\end{equation}

\begin{lemma}
Assume that the following sets are disjoint:
\begin{align*}
\{(s,n_j):s>x_j\},& \ 1\leq j\leq k'\\
 \{(x_l+\delta_{n_l}+\delta_{n_l+1}\ldots+\delta_{p-1},p): n_l+1\leq p\leq n_l'\},&\ k'+1\leq l\leq k.
\end{align*}
Then
\begin{equation}\label{LemmaToProve}
\eqref{EExpression}=\displaystyle\sum_{s_1>x_1}\cdots\sum_{s_{k'}>x_{k'}}\sum_{p_{k'+1}=n_{k'+1}+1}^{n_{k'+1}'}\cdots\sum_{p_k=n_k+1}^{n_k'}\det\left[
\begin{array}{cc}
A_{11} &  A_{12} \\
A_{21} & A_{22} 
\end{array}
\right],
\end{equation}
where the matrix blocks are:
\begin{align*}
A_{11} &= [(1-\delta_{ij})\tilde{K}(s_i,n_i,s_j,n_j,t)]_{1\leq i,j\leq k'} \\
A_{12} &= [\tilde{K}(s_i,n_i,x_j,p_j,t)]_{1\leq i\leq k',\ k'+1\leq j\leq k} \\
A_{21} &= [-\tilde{K}(x_i-\delta_{p_i-1},p_i-1,s_j,n_j,t)]_{k'+1\leq i\leq k,\ 1\leq j\leq k'}\\
A_{22} &= [-(1-\delta_{ij})\tilde{K}(x_i-\delta_{p_i-1},p_i-1,x_j,p_j,t)]_{k'+1\leq i,j\leq k}
\end{align*}
\end{lemma} 
\begin{proof}
By applying Corollary \ref{Alternative} to \eqref{HeightFormula1} and \eqref{HeightFormula3}, we see that 
\[
\mathbb{E}\left(\displaystyle\prod_{j=1}^{k'} h(x_j,n_j)\prod_{l=k'+1}^k H_{n_l,n_l'}(x_l)\right)
\]
equals the right hand side of \eqref{LemmaToProve} with the $(1-\delta_{ij})$ terms removed. It is well-known that subtracting the expectation corresponds to putting zeroes on the diagonal. For example, this is noticed in the proof of Theorem 7.2 of \cite{kn:K}.
\end{proof}

Write the determinant in \eqref{LemmaToProve} as a sum over permutations $\sigma$ in $S_k$. If the cycle decomposition of $\sigma$ contains the cycle $(c_1\ c_2\ \ldots\ c_r)$ of length $r$ and $M$ denotes the matrix in the right hand side of \eqref{LemmaToProve}, then the contribution from $\sigma$ is
\[
\displaystyle\sum_{s_1}\cdots\sum_{s_{k'}}\sum_{p_{k'+1}}\cdots\sum_{p_k}\operatorname{sgn}(\sigma)M_{c_1c_2}M_{c_2c_3}\ldots M_{c_rc_1}(\cdots)(\cdots),
\]
where $(\cdots)(\cdots)$ correspond to other cycles of $\sigma$. Let $\psi_{c_{\iota}}$ denote $s_{c_{\iota}}$ if $1\leq c_{\iota}\leq k'$, and $p_{c_{\iota}}$ if $k'<c_{\iota}\leq k$. Since the sum over $\psi_{c_{\iota}}$ only affects the matrix terms $M_{c_{\iota-1}c_{\iota}}$ and $M_{c_{\iota}c_{\iota+1}}$, the contribution from $\sigma$ is
\begin{equation}\label{ProductofCycles}
\left((-1)^{r-1}\displaystyle\sum_{\psi_{c_1}}\cdots\sum_{\psi_{c_r}}M_{c_1c_2}M_{c_2c_3}\ldots M_{c_rc_1}\right)\big(\ldots\big),
\end{equation}
where $(\ldots)$ denote other cycles. In other words, the contribution from $\sigma$ can be expressed as a product over the cycles in the cycle decomposition of $\sigma$.

Note that if $\sigma$ fixes any points,  then the correponding contribution is zero because all the diagonal entries are zero. 

\subsection{Analysis steps in proof of Theorem \ref{MainTheorem}}\label{AnalysisGeneral}

In \eqref{LemmaToProve}, set $x_j=[N\nu_j], n_l=[N\eta_l],$ and $t=N\tau$.  Our goal is to find the limit of \eqref{LemmaToProve} as $N\rightarrow\infty$.  Expanding the determinant into a sum over $\sigma\in S_k$, we just saw that the contribution from a fixed $\sigma$ is of the form \eqref{ProductofCycles}. First note that if any of the $\psi_{c_i}$ denotes $p_{c_i}$, then
\[
\displaystyle\sum_{\psi_{c_1}}\cdots\sum_{\psi_{c_r}}M_{c_1c_2}M_{c_2c_3}\ldots M_{c_rc_1}\rightarrow 0.
\]
This is because each $M_{c_jc_{j+1}}$ is proportional to $1/N$ (by Proposition \ref{Prop6.9}, so $M_{c_1c_2}M_{c_2c_3}\ldots M_{c_rc_1}$ is proportional to $N^{-r}$, but the sum is only taken over $\mathcal{O}(N^{r-1})$ terms. Therefore, \eqref{Expectation} can be expressed as a single term of the form in \eqref{EExpression}, and in this term $k'=k$. 

Now we will prove (stated as Theorem \ref{GoalinAnalysisSection} below) that
\[
\displaystyle\sum_{s_{c_1}}\cdots\sum_{s_{c_r}}M_{c_1c_2}M_{c_2c_3}\ldots M_{c_rc_1}\rightarrow
\left(\frac{1}{2\pi}\right)^r \int_{\bar{\Om_1}}^{\Om_1} dz_1 \cdots \int_{\bar{\Om_r}}^{\bar{\Om_r}} dz_r \frac{f(z_1,z_2)}{G'_{\nu}(z_1)} \ldots \frac{f(z_r,z_r)}{G'_{\nu}(z_r)}
\]
Once this is proven, \eqref{Cycle} implies that the total contribution from $S_k-\mathcal{F}_k$ equals zero. When $l=2$, then the right hand side is just $\mathcal{G}(\Om_1,\Om_2)$, completing the proof of Theorem \ref{MainTheorem}.



Recall the definitions of $G$ and $\Om$ from section \ref{Statement of the Main Theorem}. Set $\theta:\mathcal{D}\rightarrow [0,\pi)$ to be 
\[
\theta(\nu,\eta,\tau)=\frac{1}{2}\arg G''(\nu,\eta,\tau,\Omega(\nu,\eta,\tau)).
\]





\begin{proposition}\label{ProductOfKernels}
For $i=1,2,3$, let $(\nu_i,\eta_i,\tau)\in\mathcal{D}$, $x_i=[N\nu_i]$, $n_i=[N\eta_i]$ and $t=N\tau$. For $i=1,3$, let $G_i(z)$ denote $G(\nu_i,\eta_i,\tau,z)$, let $\theta_i$ denote $\theta(\nu_i,\eta_i,\tau)$ and let $\Om_i$ denote $\Om(\nu_i,\eta_i,\tau)$. Let $\Gamma_+:=\{\Om(\nu,\eta_2,\tau): \nu_2\leq \nu < q_2(\eta_2,\tau)\}$ and $\Gamma_-=\bar{\Gamma}_+$. Let $G_{\nu}'(z)=(\partial^2/\partial z\partial v)G(\nu_2,\tau_2,\tau,z)$. Then
\begin{multline}\label{KernelProduct}
\displaystyle\sum_{y>[N\nu_2]} K(x_1,n_1,y,n_2,t)K(y,n_2,x_3,n_3,t) \\
= o\left(\frac{1}{N}\right)+\frac{e^{N\Re((G_1(\Om_1)-G_3(\Om_3)))}}{2\pi N\sqrt{\vert G_1''(\Om_1)\vert}\sqrt{\vert G_3''(\Om_3)\vert}}\int_{\Gamma_+\cup\Gamma_-}\frac{dz}{2\pi G'_{\nu_1}(z)}\\
\times \Big[ f(\Om_1,z)f(z,\Om_3)\frac{e^{iN\Im(G_1(\Om_1))-i{\theta}_1}}{e^{iN\Im(G_3(\Om_3))+i{\theta}_3}} + f(\bOm_1,z)f(z,\Om_3)\frac{e^{-iN\Im(G_1(\Om_1))+i{\theta}_1}}{e^{iN\Im(G_3(\Om_3))+i{\theta}_3}}\\
+ f(\Om_1,z)f(z,\bOm_3)\frac{e^{iN\Im(G_1(\Om_1))-i{\theta}_1}}{e^{-iN\Im(G_3(\Om_3))-i{\theta}_3}} + f(\bOm_1,z)f(z,\bOm_3)\frac{e^{-iN\Im(G_1(\Om_1))+i{\theta}_1}}{e^{-iN\Im(G_3(\Om_3))-i{\theta}_3}}\Big].
\end{multline}
\end{proposition}
\begin{proof}
Let $G_2(z)$ denote $G([y/N],\eta_2,\tau,z)$, let $\theta_2$ denote $\theta([y/N],\eta_2,\tau)$ and $\Om_2$ denote $\Om(y/N,\eta_2,\tau)$. Fix some $\beta\in (-1/2,0)$ and split up the sum into two parts: the first part is from $\lf N\nu_2 \rf$ to $\lf N(q_2-N^{\beta}) \rf$, while the second sum is from $\lf N(q_2-N^{\beta}) \rf$ to $\lf Nq_2 \rf$. Since there are no particles to the right of $Nq_2$ in the limit $N\rightarrow\infty$, the sum from $Nq_2$ to $\infty$ can be ignored. It is common to refer to the first sum as the bulk and the second sum as the edge. First examine the bulk. By Proposition \ref{Prop6.9}, 
\begin{multline}
\displaystyle K(x_1,n_1,y,n_2,t)K(y,n_2,x_3,n_3,t)\\
=\frac{e^{N\Re((G_1(\Om_1)-G_2(\Om_2)))}}{2\pi N\sqrt{\vert G_1''(\Om_1)\vert}\sqrt{\vert G_2''(\Om_2)\vert}} \frac{e^{N\Re((G_2(\Om_2)-G_3(\Om_3)))}}{2\pi N\sqrt{\vert G_2''(\Om_2)\vert}\sqrt{\vert G_3''(\Om_3)\vert}}\\
\times \Big[f(\Om_1,\Om_2)f(\Om_2,\Om_3)\frac{e^{iN\Im(G_1(\Om_1))-i{\theta}_1}}{e^{iN\Im(G_2(\Om_2))+i{\theta}_2}}\frac{e^{iN\Im(G_2(\Om_2))-i{\theta}_2}}{e^{iN\Im(G_3(\Om_3))+i{\theta}_3}} + \circlearrowleft \Big]\\
+\mathcal{O}(G_2''(\Om_2)^{-4}N^{-3}) +\mathcal{O}(G_2''(\Om_2)^{-7}N^{-4}),
\end{multline}
where $\circlearrowleft$ denotes the other fifteen terms that occur in the sum. First let us examine the error term in the bulk.

By (2) of Definition \ref{Normal}, each term in the error is bounded by $(N^{\beta/2})^{-4}N^{-3}$ and $(N^{\beta/2})^{-7}N^{-4}$, respectively. There are $\sim N$ terms, and since $\beta>-1/2$, we must have $-2\beta-3+1<-1$ and $-7\beta/2-4+1<-1$. Therefore the sum is $o(1/N)$.  

Now let us return to the main term in the bulk. For eight of the sixteen terms in $\circlearrowleft$, the expression $e^{iN\Im(G_2(\Om_2))}$ cancels in the numerator and the denominator. By Proposition \ref{OneProposition}, these eight terms are $o(1/N)$. By Proposition \ref{RiemannApproximation}, the other eight terms equal 

\begin{multline*}
 \frac{e^{N\Re((G_1(\Om_1)-G_3(\Om_3)))}}{2\pi N\sqrt{\vert G_1''(\Om_1)\vert}\sqrt{\vert G_3''(\Om_3)\vert}}\int_{\nu_2}^{\infty} \frac{e^{-2i\theta_2}}{2\pi\vert G_2''(\Om_2)\vert}\\
\times \Big[f(\Om_1,\Om_2)f(\Om_2,\Om_3)\frac{e^{iN\Im(G_1(\Om_1))-i{\theta}_1}}{e^{iN\Im(G_3(\Om_3))+i{\theta}_3}} + \ldots \Big] d\nu+o\left(\frac{1}{N}\right),
\end{multline*}
where $\ldots$ represent the other seven terms. Of the eight total terms, four have $f(\cdot,\Om_2)f(\Om_2,\cdot)$ and four have $f(\cdot,\bar{\Om_2})f(\bar{\Om_2},\cdot)$. For the four terms with the expression $\Om_2$, make the substitution $z=\Om(\nu,\eta_2,\tau)$. The new integration path is $\Gamma_+$. By taking the partial of \eqref{DefinitionOfOmega} with respect to $\nu$ and using the chain rule, 
\begin{equation*}
\frac{\partial\Omega}{\partial\nu}=-\frac{G'_{\nu}(\Omega)}{G''(\Omega)},
\end{equation*}
which implies
\[
\frac{e^{-2i\theta_2}}{2\pi\vert G_2''(\Om_2)\vert}d\nu = \frac{d\nu}{2\pi G_2''(\Om_2)}=-\frac{dz}{2\pi G'_{\nu}(z)}. 
\]
For the four terms with $\bar{\Om}_2$, make the substitution $z=\bar{\Om}(\nu,\eta_2,\tau)$. The path of integration is $\Gamma_-$. Finally, the integral becomes 
\begin{multline*}
o\left(\frac{1}{N}\right)+\frac{e^{N\Re((G_1(\Om_1)-G_3(\Om_3)))}}{2\pi N\sqrt{\vert G_1''(\Om_1)\vert}\sqrt{\vert G_3''(\Om_3)\vert}}\int_{\Gamma_+\cup\Gamma_-}\frac{dz}{2\pi G'_{\nu_1}(z)}\\
\times \Big[ f(\Om_1,z)f(z,\Om_3)\frac{e^{iN\Im(G_1(\Om_1))-i{\theta}_1}}{e^{iN\Im(G_3(\Om_3))+i{\theta}_3}} + f(\bOm_1,z)f(z,\Om_3)\frac{e^{-iN\Im(G_1(\Om_1))+i{\theta}_1}}{e^{iN\Im(G_3(\Om_3))+i{\theta}_3}} \\
+ f(\Om_1,z)f(z,\bOm_3)\frac{e^{iN\Im(G_1(\Om_1))-i{\theta}_1}}{e^{-iN\Im(G_3(\Om_3))-i{\theta}_3}} + f(\bOm_1,z)f(z,\bOm_3)\frac{e^{-iN\Im(G_1(\Om_1))+i{\theta}_1}}{e^{-iN\Im(G_3(\Om_3))-i{\theta}_3}}\Big].
\end{multline*} 

Now we sum over the edge. By Proposition \ref{EdgeAsymptotics} and (2) of Definition \ref{Normal}, the sum is bounded above by 
\[
\sum_{y=(q_2-N^{\beta})N}^{q_2N} \vert G_2(\Om_2)^{-1}\vert N^{-2} \leq \sum_{y=0}^{N^{\beta+1}} \left(\frac{y}{N}\right)^{1/2}N^{-2}=\bigO(N^{3\beta/2-1}).
\]
As long as $\beta<0$, the sum over the edge is also $o(1/N)$. 
\end{proof}

\begin{theorem}\label{GoalinAnalysisSection}
For $i=1,\ldots,l$, let $(\nu_i,\eta_i,\tau)\in\mathcal{D}$ and set $x_i=[N\nu_i],n_i=N\eta_i$. For $i=1,\ldots,l$, let $G_i(z)$ denote $G(\nu_i,\eta_i,z)$, let $\theta_i$ denote $\theta(\nu_i,\eta_i,\tau)$ and let $\Om_i$ denote $\Om(\nu_i,\eta_i,\tau)$. Let $\Gamma_i^{+}:=\{\Om(\nu,\eta_i,\tau): \nu_1\leq \nu < q_2(\eta_i,\tau)\}$ and $\Gamma_i^{-}=\bar{\Gamma}_i^{+}$. Then
\begin{multline*}
\displaystyle\sum_{y_1>[N\nu_1]} \cdots \displaystyle\sum_{y_l>[N\nu_l]} \prod_{i=1}^l K(y_i,x_i,y_{i+1},x_{i+1},t) \\
\rightarrow \left(\frac{1}{2\pi}\right)^l \int_{\Gamma_1^{+}\cup\Gamma_1^{-}} dz_1 \cdots \int_{\Gamma^{+}_l\cup\Gamma^{-}_l} dz_l \frac{f(z_1,z_2)}{G'_{\nu}(z_1)} \ldots \frac{f(z_l,z_1)}{G'_{\nu}(z_l)}.
\end{multline*}
The indices are taken cyclically.
\end{theorem}
\begin{proof}
By Proposition \ref{Prop6.9}, the product has $4^l$ terms. Each application of Proposition \ref{ProductOfKernels} decreases the number of terms by a factor of 4, so repeated applications of Proposition \ref{ProductOfKernels} yields the result.
\end{proof}



\section{Specific Results}\label{Specific Results}
\subsection{Particle system with a wall}
We now return to the particle system with a reflecting wall described in the Introduction. For notational reasons, it is more convenient to use different co-ordinates. Instead of labeling the levels as $1,2,3,\ldots$, it is more convenient to label them as $(1,-1/2),(1,1/2),(2,-1/2),(2,1/2),\ldots$. If the $(n_1,a_1)$ is at least as high as the $(n_2,a_2)$ level, then this will be denoted as $(n_1,a_1)\trianglerighteq (n_2,a_2)$. This happens if and only if $2n_1+a_1\geq 2n_2+a_2$. Using the notation of Section \ref{AlgebraGeneral}, $m_{(n,a)}=n$ and $\delta_{(n,a)}=a+1/2$. Along the horizontal direction, we will use a square lattice, so that the particles live on $\mathbb{N}$ instead of $2\mathbb{N}$ or $2\mathbb{N}+1$.

Let $m_{a_1}(dz)$ be defined by
\[
m_{a_1}(dz)
\begin{cases}
\dfrac{dz}{2iz},\ & a_1=-1/2,\\
\dfrac{-(z^{1/2}-z^{-1/2})^2dz}{4iz}, \ & a_1=1/2.
\end{cases}
\]
Let $\mathsf{J}_s^{(\pm 1/2,-1/2)}$ denote the (normalized) Jacobi polynomial with parameters $(\pm 1/2,-1/2)$. The normalization is set so that for any nonzero complex number $z$, $\mathsf{J}_s^{(\pm 1/2,-1/2)}$ satisfies 
\begin{eqnarray}\label{JacobiEquations}
\mathsf{J}_s^{(-1/2,-1/2)}\left(\frac{z+z^{-1}}{2}\right)=\frac{z^s+z^{-s}}{2},\\
\mathsf{J}_s^{(1/2,-1/2)}\left(\frac{z+z^{-1}}{2}\right)=\frac{z^{s+1/2}-z^{-s-1/2}}{z^{1/2}-z^{-1/2}}.
\label{JacobiEquations2}
\end{eqnarray}

Let $W^{(a,-1/2)}(s)$ be defined for nonnegative integers $s$ by 
\[
W^{(a,-1/2)}(s)=
\begin{cases}
2,\ \ \text{if}\ \ s>0,a=-\frac{1}{2},\\
1,\ \ \text{if}\ \ s=0,a=-\frac{1}{2},\\
1,\ \ \text{if}\ \ s\geq 0,a=\frac{1}{2}.
\end{cases}
\]
Note that for $a=\pm 1/2$,
\begin{equation}\label{Orthogonality}
\frac{W^{(a,-1/2)}(s_1)}{\pi}\oint_{\vert z\vert=1}\mathsf{J}_{s_1}^{(a,-1/2)}\left(\frac{z+z^{-1}}{2}\right)\mathsf{J}_{s_2}^{(a,-1/2)}\left(\frac{z+z^{-1}}{2}\right)m_{a}(dz)=\delta_{s_1s_2}
\end{equation}

By Theorem 4.1 of \cite{kn:BK}, the correlation functions are determinantal with kernel
\begin{multline}\label{TheDoubleIntegral'}
K(n_1,a_1,s_1,n_2,a_2,s_2,t)\\
=\frac{W^{(a_1,-1/2)}(s_1)}{2\pi^2 i}\oint\oint\frac{e^{t(\frac{z+z^{-1}}{2})}}{e^{t(\frac{v+v^{-1}}{2})}} \mathsf{J}_{s_1}^{(a_1,-1/2)}\left(\frac{z+z^{-1}}{2}\right)\mathsf{J}_{s_2}^{(a_2,-1/2)}\left(\frac{v+v^{-1}}{2}\right)\\
\times\frac{(\frac{z+z^{-1}}{2}-1)^{n_1}}{(\frac{v+v^{-1}}{2}-1)^{n_2}} \frac{1-v^{-2}}{z+z^{-1}-v-v^{-1}} m_{a_1}(dz) dv\\
\end{multline}
\begin{multline}\label{ExtraIntegral'}
+\mathbf{1}_{(n_1,a_1)\trianglerighteq(n_2,a_2)}\Bigg(\frac{W^{(a_1,-1/2)}(s_1)}{\pi}\oint\mathsf{J}_{s_1}^{(a_1,-1/2)}\left(\frac{z+z^{-1}}{2}\right)\mathsf{J}_{s_2}^{(a_2,-1/2)}\left(\frac{z+z^{-1}}{2}\right)\\
\times\left(\frac{z+z^{-1}}{2}-1\right)^{n_1-n_2}m_{a_1}(dz)\Bigg),
\end{multline}
where the $z$-contour is the unit circle and the $v$-contour is a circle centered at the origin with radius bigger than $1$.

\begin{theorem}\label{IsNormal}
The determinantal point process is normal. The Green's function is given by 
\[
\mathcal{G}(z,w)=\frac{1}{2\pi}\log\left(\frac{z+z^{-1}-\bar{w}-\bar{w}^{-1}}{z+z^{-1}-w-w^{-1}}\right). 
\]
\end{theorem}
Once we prove the point process is normal, the expression for the Green's function follows from Theorem \ref{MainTheorem} with
\[
G(\nu,\eta,\tau;u)=\tau\frac{u+u^{-1}}{2}+\eta\log\left(\frac{u+u^{-1}}{2}-1\right)-\nu\log u,
\]
\[
f(u,v)=\frac{1}{v}\frac{1-u^{-2}}{v+v^{-1}-u-u^{-1}}.
\]
In section \ref{SpecificAlgebra}, we show that the third condition in Definition \ref{Normal} is satisfied.  In section \ref{SpecificAnalysis}, we show that the fourth and second conditions are satisfied. Since these are conditions are the hardest to prove, we will focus mainly on their proofs. The fifth conditions follows from the substitution $u_j=z_j+z_j^{-1}$ and (5) of Remark \ref{NormalExplanation}.  

\subsection{Algebraic steps in proof of theorem \ref{IsNormal}}\label{SpecificAlgebra}
\begin{proposition}
Let $\mathcal{C}_0(n,a,s)$ equal
\[
\mathcal{C}_0(n,a,s)=
\begin{cases}
(-1)^s (-2)^{n-1}, \ \ a=-1/2\\
(-1)^s (-2)^n, \ \ a=1/2 
\end{cases}
\]
and $c_0(n_1,a_1,s_1,n_2,a_2,s_2) =\mathcal{C}_0(n_1,a_1,s_1)/\mathcal{C}_0(n_2,a_2,s_2)$. Then $\tilde{K}=c_0K$ satisfies \eqref{C1}--\eqref{C5} for $L=1$.
\end{proposition}
\begin{proof}
Using \eqref{JacobiEquations}--\eqref{JacobiEquations2} and the orthogonality relation \eqref{Orthogonality}, it is straightforward to check that \eqref{C1} and \eqref{C2} hold. What happens is that in the left hand side of \eqref{C1} or \eqref{C2}, one obtains six terms, three of which come from \eqref{TheDoubleIntegral'} and three of which come from \eqref{ExtraIntegral'}. The three terms from \eqref{TheDoubleIntegral'} always sum to $0$, while the three terms from \eqref{ExtraIntegral'} sum to $0$ or $1$. 

Now we will prove \eqref{C6}-\eqref{C5} when $a_1=-1/2$. The term \eqref{ExtraIntegral'} equals zero, so we only need to look at \eqref{TheDoubleIntegral'}. Explicitly, the expression is
\begin{multline*}
K(n,-1/2,s,n,1/2,s',t)=\frac{2}{2\pi^2 i}\oint\oint_{\vert z\vert=1}\frac{e^{t(\frac{z+z^{-1}}{2})}}{e^{t(\frac{v+v^{-1}}{2})}} \left(\frac{z^{s}+z^{-s}}{2}\right)\\
\times\left(\frac{v^{s'+1/2}-v^{-s'-1/2}}{v^{1/2}-v^{-1/2}}\right)\frac{(\frac{z+z^{-1}}{2}-1)^{n}}{(\frac{v+v^{-1}}{2}-1)^{n}} \frac{1-v^{-2}}{z+z^{-1}-v-v^{-1}}\frac{dzdv}{2iz}, 
\end{multline*}
and we want the asymptotic result when $s,s'\rightarrow\infty$ in such a way that $s-s'$ is $0$ or $1$. Expand the paranthetical expression $v^{s'+1/2}-v^{-s'-1/2}$ to get two terms, each of which is a double integral. Since $1=\vert z\vert<\vert v\vert$, the term with $v^{-s'-1/2}$ goes to zero. For the remaining term, expand $z^s+z^{-s}$ to get two terms. For the term with $z^s$, make the substitution $z\mapsto z^{-1}$. What remains is
\[
\frac{2}{2\pi^2 i}\oint\oint_{\vert z\vert=1}\frac{e^{t(\frac{z+z^{-1}}{2})}}{e^{t(\frac{v+v^{-1}}{2})}} \frac{v^{s'}}{z^{s}}\frac{v}{v-1}\frac{(\frac{z+z^{-1}}{2}-1)^{n}}{(\frac{v+v^{-1}}{2}-1)^{n}} \frac{1-v^{-2}}{z+z^{-1}-v-v^{-1}}\frac{dzdv}{2iz}. 
\]
Now deform the $z$-contour to the circle $\vert z\vert=1+2\epsilon$ and the $v$-contour to the circle $\vert v\vert=1+\epsilon$, where $\epsilon>0$. With these deformations, $\vert v\vert<\vert z\vert$, so the double integral goes to zero. However, residues are picked up when the contours pass through each other. These residues equal
\[
-\frac{2}{\pi}\oint_{\vert z\vert=1+2\epsilon} z^{s'-s}\frac{z}{z-1}\frac{dz}{2iz}.
\]
There is a residue at $z=1$ which equals $-2$, and a residue at $z=0$ which equals $0$ for $s'\geq s$ and $2$ for $s>s'$. Since $c_0(n,-1/2,s,n,1/2,s)=-1/2$, this proves \eqref{C6} and \eqref{C5} when $a_1=-1/2$. The case when $a_1=1/2$ is similar. 

It remains to show \eqref{C3} and \eqref{C4}. When considering the product of two kernels, we obtain a quadruple integral. After the substitutions $z_1\mapsto z^{-1}_1$ and $v_2\mapsto v_2^{-1}$, the part of the integrand that depends on $s$ is just $(z_1/v_2)^s$. Therefore, deforming contours so that $\vert v_2\vert>\vert z_1\vert$ gives \eqref{C3} and \eqref{C4}.
\end{proof}

\subsection{Analysis steps in proof of theorem \ref{IsNormal}}\label{SpecificAnalysis}
For this section, we need a slightly different expression for the kernel. By (40)--(42) of \cite{kn:BK}, the kernel equals
\begin{multline}\label{TheDoubleIntegral}
K(n_1,a_1,s_1;n_2,a_2,s_2,t)\\
=\frac{W^{(a_1,-1/2)}(s_1)}{2\pi^2 i}\int_{e^{-i\theta}}^{e^{i\theta}}\oint_{\vert z\vert=1}\frac{e^{t(\frac{z+z^{-1}}{2})}}{e^{t(\frac{v+v^{-1}}{2})}} \mathsf{J}_{s_1}^{(a_1,-1/2)}\left(\frac{z+z^{-1}}{2}\right)\mathsf{J}_{s_2}^{(a_2,-1/2)}\left(\frac{v+v^{-1}}{2}\right)\\
\times\frac{(\frac{z+z^{-1}}{2}-1)^{n_1}}{(\frac{v+v^{-1}}{2}-1)^{n_2}} \frac{1-v^{-2}}{z+z^{-1}-v-v^{-1}} m_{a_1}(dz) dv\\
\end{multline}
\begin{multline}\label{ExtraIntegral}
+\mathbf{1}_{(n_1,a_1)\trianglerighteq(n_2,a_2)}\Bigg(\frac{W^{(a_1,-1/2)}(s_1)}{\pi}\oint_{\vert z\vert=1}\mathsf{J}_{s_1}^{(a_1,-1/2)}\left(\frac{z+z^{-1}}{2}\right)\mathsf{J}_{s_2}^{(a_2,-1/2)}\left(\frac{z+z^{-1}}{2}\right)\\
\times\left(\frac{z+z^{-1}}{2}-1\right)^{n_1-n_2}m_{a_1}(dz)\Bigg)
\end{multline}
\begin{multline}\label{ExtraIntegrall}
+\Bigg(\frac{W^{(a_1,-1/2)}(s_1)}{\pi}\int_{e^{-i\theta}}^{e^{i\theta}}\mathsf{J}_{s_1}^{(a_1,-1/2)}\left(\frac{z+z^{-1}}{2}\right)\mathsf{J}_{s_2}^{(a_2,-1/2)}\left(\frac{z+z^{-1}}{2}\right)\\
\times\left(\frac{z+z^{-1}}{2}-1\right)^{n_1-n_2}m_{a_1}(dz)\Bigg),
\end{multline}
where $\theta$ is any real number, and the arc from $e^{-i\theta}$ to $e^{i\theta}$ is outside the unit circle and does not cross $(-\infty,0]$.  

Set 
\[
G(\nu,\eta,\tau,z)=\tau\frac{z+z^{-1}}{2}+\eta\log\left(\frac{z+z^{-1}}{2}-1\right)-\nu\log z
\]
By Proposition 5.1.1 of \cite{kn:BK}, we can take $\mathcal{D}$ to be
\[
\mathcal{D}=\{(\nu,\eta,\tau): \eta,\tau>0, q_1(\eta,\tau)<\nu<q_2(\eta,\tau)\},
\] 
for some explicit algebraic functions $q_1$ and $q_2$.

\begin{lemma}\label{InverseOmega} Let $\Om_{\pm}$ denote $\Om(\pm\nu,\eta,\tau)$. Then $\bOm_+\Om_- \equiv 1.$ 
\begin{proof}
In general, 
\[
G'(z)=\frac{p(z)}{r(z)},
\]
where $p$ and $r$ are 
\begin{eqnarray*}
p(z) &=& \tau + (2\eta + 2\nu - \tau)z + (2\eta -2\nu - \tau)z^2 + \tau z^3,\\
r(z) &=& 2z^2(z-1).
\end{eqnarray*}
Let $p_{\pm}(z)$ denote the polynomial $p(z)$ corresponding to $(\pm\nu,\eta,\tau)$. Note that $z^3p_+(z^{-1})=p_-(z)$. By definition, $\Om_{\pm}$ is the zero of $p_{\pm}$ that is in the upper half-plane. Therefore, $\Om_-^{-1}=\bOm_+$.
\end{proof}
\end{lemma}

Now let us return to the proof of the fourth condition in Definition \ref{Normal}. Start by examining \eqref{TheDoubleIntegral}. Expanding the parantheses, we obtain four terms corresponding to $z^{s_1}v^{s_2},z^{s_1}v^{-s_2},z^{-s_1}v^{-s_2}$, and $z^{-s_1}v^{s_2}$. For the two terms with $z^{s_1}$, make the substitution $z\rightarrow z^{-1}$. What remains are two terms, corresponding to $z^{-s_1}v^{s_2}$ and $z^{-s_1}v^{-s_2}$.  Therefore, \eqref{TheDoubleIntegral} equals
\begin{multline}\label{TheDoubleIntegral2}
\frac{W^{(-1/2,-1/2)}(s_1)}{4\pi^2 i}\int_{e^{-i\theta}}^{e^{i\theta}}\oint\frac{e^{t(\frac{z+z^{-1}}{2})}}{e^{t(\frac{v+v^{-1}}{2})}} z^{-s_1}(v^{s_2}+v^{-s_2})\\
\times \frac{(\frac{z+z^{-1}}{2}-1)^{n_1}}{(\frac{v+v^{-1}}{2}-1)^{n_2}} \frac{1-v^{-2}}{z+z^{-1}-v-v^{-1}}\frac{dz dv}{2iz},
\end{multline}

We now need to deform the contours in \eqref{TheDoubleIntegral2} to steepest descent paths. In other words, we need
\begin{equation}\label{Inequality1}
\Re(G(\nu_1,\eta_1,\tau,z))<\Re(G(\nu_1,\eta_1,\tau,\Om(\nu_1,\eta_1,\tau)))
\end{equation}
for all $z$ on the $z$-contour and
\begin{eqnarray}\label{Inequality2} 
\Re(G(\nu_2,\eta_2,\tau,v))&>&\Re(G(\nu_2,\eta_2,\tau,\Om(\nu_2,\eta_2,\tau))),\\
\Re(G(-\nu_2,\eta_2,\tau,v))&>&\Re(G(-\nu_2,\eta_2,\tau,\Om(-\nu_2,\eta_2,\tau)))
\label{Inequality3}
\end{eqnarray}
for all $v$ on the $v$-contour. By Lemma \ref{InverseOmega} and the definition of $G$, we see that $\Re(G(\nu_2,\eta_2,\tau,\Om(\nu_2,\eta_2,\tau)))=\Re(G(-\nu_2,\eta_2,\tau,\Om(-\nu_2,\eta_2,\tau)))$. If $\vert v\vert\geq1$, then $\Re(G(-\nu_2,\eta_2,\tau,v))\geq\Re(G(\nu_2,\eta_2,\tau,v))$. Since the steepest descent paths can go completely outside the unit circle (see Proposition 5.1.2 of \cite{kn:BK}), \eqref{Inequality3} follows from \eqref{Inequality2}.

\begin{figure}[htp]
\caption{On the left is $\Re(G(\nu_1,\eta_1,\tau,z)-G(\nu_1,\eta_1,\tau,\Om(\nu_1,\eta_1,\tau)))$, and on the right is $\Re(G(\nu_2,\eta_2,\tau,v)-G(\nu_2,\eta_2,\tau,\Om(\nu_2,\eta_2,\tau)))$. White regions indicate $\Re<0$ and shaded regions indicate $\Re>0$. The double zero occurs at $\Om(\nu_j,\eta_j,\tau)$. The arc $v$ goes from $e^{-i\theta}$ to $e^{i\theta}$. The unit circle has been drawn on the right.}
\includegraphics[totalheight=0.30\textheight]{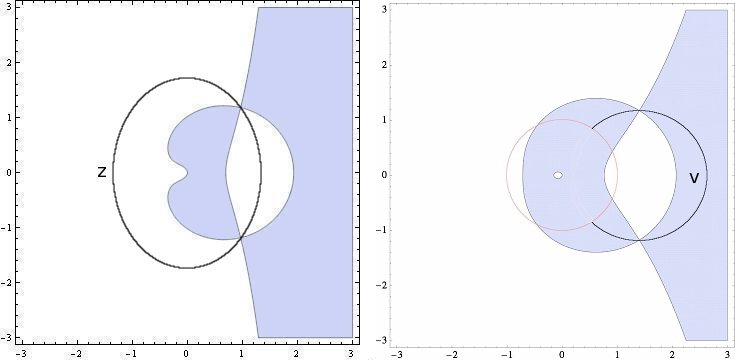}
\label{Contours}
\end{figure}

If we deform the contours to the steepest descent paths $\Gamma_1$ and $\Gamma_2$ in Figure \ref{Contours}, we get that \eqref{TheDoubleIntegral} asymptotically becomes
\[
\left(\frac{1}{2\pi i}\right)^2\int_{\Gamma_1}\int_{\Gamma_2}\frac{\exp(NG(\nu_1,\eta_1,\tau,z))}{\exp(NG(\nu_2,\eta_2,\tau,v))}\frac{1-v^{-2}}{z+z^{-1}-v-v^{-1}}\frac{dvdz}{z}
\]
\[
+\left(\frac{1}{2\pi i}\right)^2\int_{\Gamma_1}\int_{\Gamma_2}\frac{\exp(NG(\nu_1,\eta_1,\tau,z))}{\exp(NG(-\nu_2,\eta_2,\tau,v))}\frac{1-v^{-2}}{z+z^{-1}-v-v^{-1}}\frac{dvdz}{z},
\]
plus possibly the residues at $z=v$. Since $\Gamma_2$ goes outside the unit circle and the critical point of $G(-\nu_2,\eta_2,\tau,v)$ lies inside the unit circle, the second double integral is negligible.  

Now we need to compute the possible residues at $z=v$. If the contours pass through each other, then the residues at $z=v$ equal
\begin{multline}\label{Residues}
\frac{W^{(-1/2,-1/2)}(s_1)}{4\pi i}\int_{e^{i\theta}}^{\zeta}z^{s_2-s_1}\left(\frac{z+z^{-1}}{2}-1\right)^{n_1-n_2}\frac{dz}{z}\\
+\frac{W^{(-1/2,-1/2)}(s_1)}{4\pi i}\int_{\bar{\zeta}}^{e^{-i\theta}}z^{s_2-s_1}\left(\frac{z+z^{-1}}{2}-1\right)^{n_1-n_2}\frac{dz}{z}
\end{multline}
\begin{multline}
+\frac{W^{(-1/2,-1/2)}(s_1)}{4\pi i}\int_{e^{i\theta}}^{\zeta}z^{-s_2-s_1}\left(\frac{z+z^{-1}}{2}-1\right)^{n_1-n_2}\frac{dz}{z}\\
+\frac{W^{(-1/2,-1/2)}(s_1)}{4\pi i}\int_{\bar{\zeta}}^{e^{-i\theta}}z^{-s_2-s_1}\left(\frac{z+z^{-1}}{2}-1\right)^{n_1-n_2}\frac{dz}{z},
\label{Residues2}
\end{multline}
where $\zeta$ is any complex number satisfying \eqref{Inequality1} and \eqref{Inequality2}. See Figure \ref{ResidueContours}. If the contours do not pass through each other, then there is no contribution from the residues. For notational convenience, set 
\[
\xi=
\begin{cases}
\zeta,\ \text{if}\ \zeta \ \text{exists},\\
e^{i\theta},\ \text{otherwise}.
\end{cases}
\]

\begin{figure}[htp]
\begin{center}
\caption{The $z$ and $v$ contours from Figure \ref{Contours}. They intersect at $\zeta$.}
\includegraphics[totalheight=0.15\textheight]{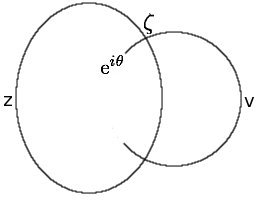}
\label{ResidueContours}
\end{center}
\end{figure}

It is important to note that $\xi$ is arbitrarily selected. The only requirement on $\zeta$ is that it satisfies the inequalities \eqref{Inequality1} and \eqref{Inequality2}, and the only requirement on $e^{i\theta}$ is that $\Re(G_2(e^{i\theta}))>\Re(G_2(\Om_2))$. So there exists $\epsilon>0$ such that  if $\vert \xi_1-\xi\vert<\epsilon$, then $\xi_1$ also satisfies those inequalities.

Now we need to compute \eqref{ExtraIntegral} and \eqref{ExtraIntegrall}. Expanding the parantheses, we get four terms corresponding to $z^{s_1+s_2},z^{s_1-s_2},z^{s_2-s_1},z^{-s_1-s_2}$. For the terms corresponding $z^{-s_2-s_1}$ and $z^{s_1-s_2}$, make the substitution $z\rightarrow z^{-1}$. Therefore, the sum of \eqref{ExtraIntegral},\eqref{ExtraIntegrall},\eqref{Residues},\eqref{Residues2} equals 
\begin{multline}\label{AnnoyingIntegral}
\frac{1}{4\pi i}\int_{\bar{\xi}}^{\xi} z^{s_2-s_1} \left(\frac{z+z^{-1}}{2}-1\right)^{n_1-n_2}\frac{dz}{z}\\
+\frac{1}{4\pi i}\int_{\bar{\xi}}^{\xi} z^{-s_2-s_1} \left(\frac{z+z^{-1}}{2}-1\right)^{n_1-n_2}\frac{dz}{z},
\end{multline}
where the contour crosses $(0,\infty)$ if $n_1\geq n_2$, and it crosses $(-\infty,0)$ if $n_1<n_2$. For each integral, deform the contour to a circular arc of constant radius. It is not a difficult calculus exercise to show that the absolute value of the integrand is maximized at the endpoints. 


Using a standard asymptotic analysis (see e.g. chapter 3 of \cite{kn:M}), we get that the asymptotic expansion of \eqref{AnnoyingIntegral} is
\begin{multline*}
\frac{c_1}{N}\xi^{N(\nu_2-\nu_1)}\left(\frac{\xi+\xi^{-1}}{2}-1\right)^{N(\eta_1-\eta_2)}+\frac{\overline{c_1}}{N}\bar{\xi}^{N(\nu_2-\nu_1)}\left(\frac{\bar{\xi}+\bar{\xi}^{-1}}{2}-1\right)^{N(\eta_1-\eta_2)}\\
+\frac{c_2}{N}\xi^{N(-\nu_2-\nu_1)}\left(\frac{\xi+\xi^{-1}}{2}-1\right)^{N(\eta_1-\eta_2)}+\frac{\overline{c_2}}{N}\bar{\xi}^{N(-\nu_2-\nu_1)}\left(\frac{\bar{\xi}+\bar{\xi}^{-1}}{2}-1\right)^{N(\eta_1-\eta_2)}
\end{multline*}
for some constants $c_1,c_2$. To complete the proof, notice that if 
\[
\left|\xi^{\pm\nu_2-\nu_1}\left(\frac{\xi+\xi^{-1}}{2}-1\right)^{\eta_1-\eta_2}\right|>e^{\Re(G_1(\Om_1)-G_2(\Om_2))}
\]
for some selection of $\pm$, then the asymptotic expansion of the kernel would depend on $\xi$. But $\xi$ was arbitrarily selected, so this is impossible. 

Now that the fourth condition has been proved, it remains to show that the second condition in Definition \ref{Normal} holds. Recall that $\Om(\nu,\eta,\tau)$ is the root of $p(\nu,\eta,\tau,z)$ that lies in the upper half-plane, where $p$ is the polynomial from Lemma \ref{InverseOmega}. We thus need to solve
\[
p(q_2(\eta,\tau)-\epsilon_1,\eta,\tau,\Om(q_2(\eta,\tau),\eta,\tau)+\epsilon_2)=0.
\]
Since $\Om(q_2(\eta,\tau),\eta,\tau)$ is a double zero of $p(q_2(\eta,\tau),\eta,\tau,z)$, we thus have to solve
\[
\frac{1}{2}\epsilon_2^2p''(\Om)-2\epsilon_1(\Om+\epsilon_2-\epsilon_2^2-2\epsilon_2\Om-\Om^2)+\bigO(\epsilon_2^3)=0,
\]
which implies that $\epsilon_2=\bigO(\epsilon_1^{1/2})$. In other words, as $\nu$ approahces $q_2(\eta,\tau)$, $\Om(\nu,\eta,\tau)-\Om(q_2(\eta,\tau),\eta,\tau)=\bigO((q_2(\eta,\tau)-\nu)^{1/2})$. Plugging this into the expression for $G''$ gives the result.

\section{Asymptotic Lemmas}
\subsection{Riemannian Approximations}
\begin{lemma}\label{ZeroLemma}
Suppose that $g\in C^1[a,b]$ and $I\in C^2[a,b]$. Suppose that as $\delta\rightarrow 0$, the Lesbesgue measure of the set $\{x\in [a,b]: I'(x)\in 2\pi\Z+[-\delta,\delta]\}$ is $\bigO(\delta^a)$ for some positive $a$. 
Let $\epsilon_N\in [-1,1]$ depend on $N$. Then
\[
\displaystyle\lim_{N\rightarrow\infty}\sum_{k=1}^{\lfloor (b-a)N \rfloor}e^{iNI(a+(k+\epsilon_N)/N)}g\left(a+\frac{k+\epsilon_N}{N}\right)\frac{1}{N} = o(1).
\]
\end{lemma}
\begin{proof}

Let $t_k$ denote $a+(k+\epsilon_N)/N$. Note that $\vert t_k-t_s\vert = \vert k-s\vert/N$. Fix some $1+N^{1/3}\leq s\leq \lfloor (b-a)N\rfloor-N^{1/3}$ and consider
\[
\displaystyle\sum_{k=s-N^{1/3}}^{s+N^{1/3}} e^{iNI(t_k)}g\left(t_k\right)\frac{1}{N}.
\]
We bound this sum in two cases.

\textbf{Case 1}. Assume $I'(t_s)\notin 2\pi\mathbb{Z}+[-\delta,\delta]$. 
For $s-N^{1/3}\leq k\leq s+N^{1/3}$, Taylor's theorem says that
\[
I(t_k)=[I(t_s)+I'(t_s)(t_k-t_s)]+[\frac{1}{2}I''(c_k)(t_k-t_s)^2]=:[I_1(t_k)]+[I_2(t_k)]
\]
for some $c_k$ between $t_s$ and $t_k$. We will prove that 
\[
\displaystyle\sum_{k=s-N^{1/3}}^{s+N^{1/3}} e^{iNI(t_k)}g\left(t_k\right)\frac{1}{N}\leq\frac{99\delta^{-1}\Vert g\Vert_{\infty}}{N}+\frac{18\Vert g\Vert_{\infty}\Vert I''\Vert_{\infty}}{N}+\frac{3\Vert g'\Vert_{\infty}}{N^{4/3}}
\]

Using the inequality
\[
\vert g(t_k)-g(t_s) \vert \leq \Vert g'\Vert_{\infty}\cdot\vert t_k-t_s\vert = \Vert g'\Vert _{\infty} \frac{\vert k-s\vert}{N},  
\]
we have that
\begin{multline}\label{E1}
\displaystyle\left|\sum_{k=s-N^{1/3}}^{s+N^{1/3}} e^{iNI(t_k)}(g\left(t_k\right)-g(t_s))\frac{1}{N}\right|\leq \sum_{k=s-N^{1/3}}^{s+N^{1/3}}\Vert g'\Vert _{\infty} \frac{\vert k-s\vert}{N^2}\\
=2\Vert g'\Vert _{\infty} \frac{N^{1/3}(N^{1/3}+1)}{N^2}.
\end{multline}
Furthermore, for $\vert k-s\vert \leq N^{1/3}$,
\begin{equation}
\label{E2}
\vert 1 - e^{iNI_2(t_k)}\vert = \vert 1- e^{iI''(c_k)(k-s)^2/(2N)}\vert \leq \vert 1- e^{i\Vert I''\Vert_{\infty}N^{-1/3}}\vert \leq 9\Vert I''\Vert_{\infty}N^{-1/3}.
\end{equation}
Also, 
\begin{multline}\label{E3}
\displaystyle\left|\sum_{k=s-N^{1/3}}^{s+N^{1/3}} e^{iNI_1(t_k)} \right|=\left|e^{iNI(t_s)}\sums e^{iI'(t_s)(k-s)}\right|\\
\leq  \left|\frac{4}{e^{iI'(t_s)}-1} \right|\leq 99\delta^{-1}
\end{multline}

Using \eqref{E1}, the definition of $I_1$ and $I_2$, \eqref{E2} and \eqref{E3} respectively, 
\begin{eqnarray*}
\left|\displaystyle\sum_{k=s-N^{1/3}}^{s+N^{1/3}} e^{iNI(t_k)}g\left(t_k\right)\frac{1}{N}\right|&\leq&\left|\displaystyle g\left(t_s\right)\sum_{k=s-N^{1/3}}^{s+N^{1/3}} e^{iNI(t_k)}\frac{1}{N}\right|+3\Vert g'\Vert_{\infty}N^{-4/3}\\
&\leq&\Vert g\Vert_{\infty}\left|\displaystyle\sum_{k=s-N^{1/3}}^{s+N^{1/3}} \frac{e^{iNI_1(t_k)}+e^{iNI_1(t_k)}(e^{iNI_2(t_k)}-1)}{N}\right|+\frac{3\Vert g'\Vert_{\infty}}{N^{4/3}}\\
&\leq& \Vert g\Vert_{\infty}\left|\sums\frac{e^{iNI_1(t_k)}}{N}\right|+\frac{18\Vert g\Vert_{\infty}\Vert I''\Vert_{\infty}}{N}+\frac{3\Vert g'\Vert_{\infty}}{N^{4/3}}\\
&\leq& \frac{99\delta^{-1}\Vert g\Vert_{\infty}}{N}+\frac{18\Vert g\Vert_{\infty}\Vert I''\Vert_{\infty}}{N}+\frac{3\Vert g'\Vert_{\infty}}{N^{4/3}}
\end{eqnarray*}

\textbf{Case 2}. Assume that $I'(t_s)\in 2\pi\mathbb{Z}+(-\delta,\delta)$. In this case, only a simple estimate is needed:
\[
\left|\sums e^{iNI(t_k)}g(t_k)\frac{1}{N}\right|\leq \frac{2\Vert g\Vert_{\infty}}{N^{2/3}}.
\]

Since the estimate in case 1 is much better than the estimate in case 2, we need an upper bound on how frequently case 2 can occur. In other words, we need an upper bound on the measure of the set $\{x\in [a,b]: I'(x)\in 2\pi\mathbb{Z}+(\delta,\delta)\}$. We assumed that
\[
\vert \{x\in [a,b]: I'(x)\in 2\pi\mathbb{Z}+[\delta,\delta]\}\vert = \bigO(\delta^a).
\]


Now we need to sum over all $s$ in the set $\{N^{1/3}+1, 3N^{1/3}+1,5N^{1/3}+1\ldots,(b-a)N-N^{1/3}\}$. There are $\bigO(\delta^a N^{2/3})$ terms for which case 2 applies. Therefore, 
\[
\sum_{k=1}^{\lfloor (b-a)N \rfloor}e^{iNI(a+(k+\epsilon_N)/N)}g\left(a+\frac{k+\epsilon_N}{N}\right)\frac{1}{N} = \mathcal{O}\left(\delta^{-1}N^{-1/3}\right)+\mathcal{O}(\delta^a),
\]
and setting $\delta=N^{-1/6}$ yields the result.
\end{proof}

\begin{proposition}\label{OneProposition}
Suppose $f:\mathbb{Z}_{\geq 0}\rightarrow \mathbb{C}$ is a function such that for each $t>0$,
\[
f(\lfloor tN \rfloor)=e^{iNI(t)}g(t)N^d + o(N^d)\ \text{as}\ N\rightarrow\infty,
\]
where $g$ and $I$ satisfy the same assumptions as in Lemma \ref{ZeroLemma}. Further suppose that the error term $o(N^d)$ is uniform, i.e.
\[
\frac{f(\lfloor tN \rfloor)-e^{iNI(t)}g(t)N^d}{N^d}\rightarrow 0\ \text{uniformly on}\ [a,b].
\]
Then as $N\rightarrow\infty$,
\[
\displaystyle\sum_{x=\lfloor aN\rfloor +1}^{\lfloor bN\rfloor} f(x) = o(N^{d+1}).
\]
\end{proposition}
\begin{proof}
This follows quickly from Lemma \ref{ZeroLemma}.
\end{proof}

\begin{proposition}\label{RiemannApproximation}
Suppose $f:\mathbb{Z}_{\geq 0}\rightarrow \mathbb{R}$ is a function such that for each $t>0$,
\[
f(\lfloor tN \rfloor)=g(t)N^d + o(N^d)\ \text{as}\ N\rightarrow\infty,
\]
where $g$ is a function on $[a,b]$ of bounded variation. Further suppose that the error term $o(N^d)$ is uniform, i.e.
\[
\frac{f(\lfloor tN \rfloor)-g(t)N^d}{N^d}\rightarrow 0\ \text{uniformly on}\ [a,b].
\]
Then
\[
\displaystyle\sum_{x=\lfloor aN\rfloor +1}^{\lfloor bN\rfloor} f(x) = N^{d+1}\int_a^b g(t)dt + o(N^{d+1}).
\]
\end{proposition}
\begin{proof}
This is an elementary, albeit somewhat tedious, exercise in approximating integrals with Riemann sums.
\end{proof}

\subsection{Asymptotics}\label{Asymptotics}
\begin{proposition}\label{Prop6.9}
For $j=1,2$, let $(\nu_j,\eta_j,\tau)\in\mathcal{D}$, $\Om_j$ denote $\Om(\nu_j,\eta_j,\tau)$, $G_j(z)$ denote $G(\nu_j,\eta_j,\tau,z)$, and $\theta_j$ denote $\theta(\nu_j,\eta_j,\tau)$. With the assumptions from section \ref{Statement of the Main Theorem}, 
\begin{multline*}
\left(\frac{1}{2\pi i}\right)^2\int_{\Gamma_1}\int_{\Gamma_2}\frac{\exp(NG(\eta_1,\nu_1,\tau,u))}{\exp(NG(\eta_2,\nu_2,\tau,w))}f(u,w)dwdu\\
= \bigO\left(\frac{G_1''(\Om_1)^{-3}+G_2''(\Om_2)^{-3}}{G_1''(\Om_1)^{1/2}G_2''(\Om_2)^{1/2}}N^{-2}\right)+\bigO(G_1''(\Om_1)^{-7/2}G_2''(\Om_2)^{-7/2}N^{-3})\\ +\frac{e^{N\Re((G_1(\Om_1)-G_2(\Om_2)))}}{2\pi N\sqrt{\vert G_1''(\Om_1)\vert}\sqrt{\vert G_2''(\Om_2)\vert}}
\times \Big[f(\Om_1,\Om_2)\frac{e^{iN\Im(G_1(\Om_1))-i{\theta}_1}}{e^{iN\Im(G_2(\Om_2))+i{\theta}_2}} + f(\Om_1,\bar{\Om}_2)\frac{e^{iN\Im(G_1(\Om_1))-i{\theta}_1}}{e^{-iN\Im(G_2(\Om_2))-i{\theta}_2}} \\
+f(\bar{\Om}_1,\Om_2)\frac{e^{-iN\Im(G_1(\Om_1))+i{\theta}_1}}{e^{iN\Im(G_2(\Om_2))+i{\theta}_2}}+f(\bar{\Om}_1,\bar{\Om}_2)\frac{e^{-iN\Im(G_1(\Om_1))+i{\theta}_1}}{e^{-iN\Im(G_2(\Om_2)-i{\theta}_2)}}\Big].
\end{multline*}
\end{proposition}
\begin{proof}
First, we show that the main term is correct.

By assumption, we can deform $\Gamma_1$ and $\Gamma_2$ so that $\Gamma_j$ passes through $\Om_j,\bar{\Om}_j$ for $j=1,2$. The contributions to the integral away from $\Om_j,\bar{\Om}_j$ are exponentially small, so we can replace $\Gamma_j$ with $\gamma_j \cup \bar{\gamma}_j$, where $\gamma_j$ and $\bar{\gamma}_j$ are steepest descent paths near $\Om_j$ and $\bar{\Om}_j$, respectively. The integration over $u\in\gamma_1\cup\bar{\gamma}_1,w\in\gamma_2\cup\bar{\gamma}_2$ expands into four integrations corresponding to $(u,w)\in\gamma_1\times\gamma_2, \bar{\gamma}_1\times\gamma_2, \gamma_1\times\bar{\gamma}_2, \bar{\gamma}_1\times\bar{\gamma}_2$. We explicitly do the calculation for $\gamma_1\times\gamma_2$. The other three calculations are essentially identical. 

Make the substitutions $s=G_1(\Om_1)-G_1(u)$ and $t=G_2(\Om_2)-G_2(w)$. In the neighborhood of $u=\Om_1$ and $w=\Om_2$, we have
\[
f(u,w)\approx f(\Om_1,\Om_2),\ \ s=-\frac{(u-\Om_1)^2}{2}G_1''(\Om_1), \ \ t=-\frac{(w-\Om_2)^2}{2}G_2''(\Om_2),
\]
which imply
\begin{align*}
G_1'(u)&=-\frac{ds}{du}=(u-\Om_1)G_1''(\Om_1)=\sqrt{-2sG_1''(\Om_1)}, \\  
G_2'(w)&=-\frac{dt}{dw}=(w-\Om_2)G_2''(\Om_2)=\sqrt{-2tG_2''(\Om_2)}.
\end{align*}
Then we get
\begin{align*}
&\left(\frac{1}{2\pi i}\right)^2e^{N(G_1(\Om_1)-G_2(\Om_2))}\int_0^{\infty}\int_0^{\infty}e^{-N(s+t)}\frac{f(u,w)}{G_1'(u)G_2'(w)}dtds\\
&=4\cdot\frac{e^{N(G_1(\Om_1)-G_2(\Om_2))}}{8\pi^2\sqrt{G_1''(\Om_1)}\sqrt{G_2''(\Om_2)}}f(\Om_1,\Om_2)\Big(\int_0^{\infty} s^{-1/2}e^{-Ns}ds\Big)\Big(\int_0^{\infty} t^{-1/2}e^{-Nt}dt\Big)\\
&=\frac{e^{N(G_1(\Om_1)-G_2(\Om_2))}}{2\pi N\sqrt{G_1''(\Om_1)}\sqrt{G_2''(\Om_2)}}f(\Om_1,\Om_2)\\
&=\frac{e^{N\Re((G_1(\Om_1)-G_2(\Om_2)))}}{2\pi N\sqrt{\vert G_1''(\Om_1)\vert}\sqrt{\vert G_2''(\Om_2)\vert}} \Big[f(\Om_1,\Om_2)\frac{e^{iN\Im(G_1(\Om_1))-i{\theta}_1}}{e^{iN\Im(G_2(\Om_2))+i{\theta}_2}}\Big], 
\end{align*}
where the last equality follows from $G(\bar{z})=\overline{G(z)}$. The $4$ appears because the maps $u\mapsto s$ and $w\mapsto t$ are both two-to-one. 

It still remains to show that the error term is correct. The remainder of this section is devoted to proving this. The idea is to reduce the double integral to progressively simpler forms. First, by a reparametrization, the integral over two arcs in $\C$ can be written as a integral in $\R^2$. Second, by using a Taylor approximation, the integral in $\R^2$ can be written as a product of two integrals in $\R$, each of which is of the form $\int e^{-NR(t)}\phi(t)dt$, where $R(t)$ has a maximum $\tm$ in the interval of integration. Third, by using the implicit function theorem, this integral reduces to the form $\int e^{-Nt^2}g(t)$, where the interval of integration is a small neighbourhood $\tm$. Fourth, this last integral is a slight generalization of $\int e^{-Nt}g(t)dt$, which is dealt with by the well-known Watson's lemma (Lemma \ref{4.8} below). Since the first two steps have been done before (see Chapters 3 and 4 of \cite{kn:M}), we will focus mostly on the third and fourth steps. 

\begin{lemma}\label{4.8}
Suppose that $R$ and $\phi$ are infinitely continuously differentiable in some neighbourhood of $t_{max}$. Also suppose that $\tm$ is a local maximum of $R$ and $R''(\tm)<0$. Then for any $N>1$ and $s\in [0,m^2]$,
\begin{multline*}
\left| \int_{\tm-\delta_1}^{\tm+\delta_2} e^{NR(t)}\phi(t)dt-\phi(\tm)e^{NR(\tm)}\sqrt{\frac{-2\pi}{NR''(\tm)}}\right|\leq \frac{\sqrt{\pi}}{2}\frac{\displaystyle\sup_{0\leq\tau\leq s} \vert g'(\tau)\vert}{N^{3/2}}\\
+e^{-Ns}\int_s^{m^2} \vert h(t)\vert dt + e^{-Nm^2}\int_{m}^{\max(\alpha,\beta)}\vert h(t)\vert dt+h(0)\int_s^{\infty} e^{-Nt}t^{-1/2}dt\\
+e^{-Ns/2}\sup_{0\leq\tau\leq s} \vert g'(\tau)\vert,
\end{multline*}
where
\[
\alpha=\sqrt{-R(\tm-\delta_1)},\ \beta=\sqrt{-R(\tm+\delta_2)},\ m=\min(\alpha,\beta),
\]
\[
h(s)=\phi(\tm+sv(s))(sv'(s)+v(s)),\ g(s)=\frac{1}{2}(h(s^{1/2})+h(-s^{1/2})), 
\]
where $v(s)$ is an infinitely differentiable function solving
\begin{equation}\label{Defnofv}
-R(\tm)+R(\tm+sv(s))=-s^2.
\end{equation}
\end{lemma}
\begin{proof}
This is a slight generalization of Watson's lemma (e.g. Proposition 2.1 of \cite{kn:M}), which deals with asymptotics of integrals of the form $\int_0^T e^{-Nt}\phi(t)dt$. By following pages pages 58--60 of \cite{kn:M}, one generalizes to integrals of the form $\int_{-\alpha}^{\beta} e^{-Nt^2}\phi(t)dt$, and then it is not hard to generalize to functions $R(t)$ which behave like $-t^2$ near its maximum.
\end{proof}

Before continuing, a few estimates on $v(s)$ are needed.

\begin{lemma}\label{EstimateOnV} Let $v(s)$ be as in \eqref{Defnofv}.

(a)
\begin{equation}\label{v0}
R''(\tm)=-2v(0)^{-2}.
\end{equation}
\begin{equation}\label{v'0}
v'(0)=\frac{R'''(\tm)}{12}v(0)^4
\end{equation}

(b) Set $B=\sup\vert R^{(4)}/24\vert$. Then 
\[
\vert v(s)-v(0)-v'(0)s\vert < \left(\frac{5}{144}R'''(\tm)^2v(0)^7 + Bv(0)^5\right)s^2.
\]
for 
\[
s<\min(\vert 53R'''(\tm)\vert/(625Bv(0)),(\vert R'''(\tm)\vert v(0)^3)^{-1},(50\sqrt{B}v(0)^2)^{-1}).
\]
In particular, $\vert v(s)-v(0)\vert< \vert R'''(\tm)\vert v(0)^4 \vs/4$ and $\vert v(s)-v(0)\vert<v(0)/4$.

(c) Let
\[
a_3:=\left(\frac{157}{16}Bv(0)^3 + \frac{101}{288}R'''(\tm)^2v(0)^5\right).
\]
Then
\[
\vert v(s)+sv'(s)-v(0)-2v'(0)s\vert< \left( \frac{39}{32}R'''(\tm)^2v(0)^7 + \frac{471}{16}Bv(0)^5\right) s^2
\]
\begin{equation}\label{Estimateons}
 \text{for}\ \  \vert s\vert <\min\left(\frac{\vert 53R'''(\tm)\vert}{625Bv(0)},\frac{1}{50\sqrt{B}v(0)^2},\frac{1}{6R'''(\tm)v(0)^3},\left|\sqrt{\frac{2v(0)^{-1}}{3a_3}}\right|\right). 
\end{equation}


(d) With the same bounds on $\vs$, 
\[
\vert 2v'(s)+sv''(s)-2v'(0)\vert<450000(R'''(\tm)^2v(0)^7 + Bv(0)^5)s 
\]
\end{lemma}
\begin{proof}
(a) The proof comes from page 69 of \cite{kn:M}. It follows immediately from using implicit differentiation of \eqref{Defnofv} and setting $s=0$.

(b) First notice that if $R_-(t) \leq R(t) \leq R_+(t)$ with $R_-(\tm) = R(\tm) = R_+(\tm)$ and $v_{\pm}$ are the solutions to $-R(\tm)+R(\tm \pm sv_{\pm}(s))=-s^2$, then $v_- \leq v \leq v_+$. We will use 
\[
R_{\pm}(t)=R(\tm)+\frac{1}{2}R''(\tm)(t-\tm)^2 + \frac{1}{6}R'''(\tm)(t-\tm)^3  \pm B(t-\tm)^4 
\]
Therefore, we obtain bounds on $v(s)$ by solving $-R(\tm)+R(\tm \pm sv_{\pm}(s))=-s^2$, which is equivalent to solving
\[
Q_{\pm,s}(y):=1-y_0^{-2}y^2+Asy^3\pm Bs^2y^4=0, \ \ A=R'''(\tm)/6, \ \ y_0=v(0).
\]
In other words $Q_{\pm,s}(v(s))=0$. Taking the derivative of $Q_{\pm,s}(v(s))=0$ with respect to $s$ and setting $s=0$, observe that $v'(0)=Ay_0^4s/2.$ We will use the intermediate value theorem to estimate roots of $Q_{\pm,s}$. 

For $\epsilon = Ay_0^4s/2 + (5A^2y_0^7/4 + By_0^5)s^2$ and $\vs=(HAy_0^3)^{-1}$ where $H$ is any real number, we have
\begin{multline*}
Q_{+,s}(y_0+\epsilon)\leq\frac{1}{256 A^{10} H^{10} y_0^{10}}(256 B^5 + 
  256 A^2 B^4 p_2(H) y_0^2 
  + 
  32 A^4 B^3 p_4(H) y_0^4 \\+ 
  16 A^6 B^2 p_6(H) y_0^6 
     + 
  A^8 B p_8(H) y_0^8 + 
  4 A^{10} H^3 p_5(H) y_0^{10})
\end{multline*}
where the $p_i$ are polynomials which satisfy the following inequalities when $\vert H\vert>6:$
$$ 
p_2(H)<11H^2,p_4(H)<371H^4,p_6(H)<1447H^6,p_8(H)<-10H^8,p_5(H)<0.
$$

Now setting $H := hy_0^{-1}A^{-1}\sqrt{B}$ where $h>50$,
\[
Q_{+,s}(y_0+\epsilon)<\frac{128 + 1408 h^2 + 5936 h^4 + 11576 h^6 - 5 h^8}{128h^{10}}<0.
\]
Since 
\[
Q_{+,s}(y_0+Ay_0^4s/2)>\frac{1}{16} s^2 y_0^4 (B (2 + A s y_0^3)^4 + 
   2 A^2 y_0^2 (10 + 6 A s y_0^3 + A^2 s^2 y_0^6))>0,
\] 
this implies that $v_+(s)<y_0+\epsilon=v(0)+v'(0)s + (5A^2v(0)^7/4 + Bv(0)^5)s^2$ for $\vs<\min((6\vert A\vert v(0)^3)^{-1}, (50\sqrt{B}v(0)^2)^{-1})$.

By applying a similar argument to $Q_{-,s}$, one can show that $v(0) + v'(0)s - (5A^2v(0)^7/4 + Bv(0)^5)s^2<v_-(s)<y_0$. Thus the lower bound holds in both cases. 

The last statement follows because
\begin{eqnarray*}
\vert v(s) - v(0) \vert &<& \vert v'(0)\vert\cdot \vs + \left(\frac{5}{144}R'''(\tm)^2v(0)^7 + Bv(0)^5\right)s^2 \\ 
&<& \frac{\vert R'''(\tm) \vert}{12}v(0)^4\vs + \frac{5}{144}\frac{R'''(\tm)^2v(0)^7}{\vert R'''(\tm)\vert v(0)^3}\vs + Bv(0)^5\frac{53 \vert R'''(\tm)\vert }{625Bv(0)}\vs\\
&<& \frac{\vert R'''(\tm)\vert v(0)^4}{4}\vs < \frac{v(0)}{4}.
\end{eqnarray*}

(c) Differentiating \eqref{Defnofv} yields
\[
(v(s)+sv'(s))=\frac{-2s}{R'(\tm+sv(s))}.
\]
To estimate this, let us first estimate $R'$.

By a Taylor expansion,
\begin{equation}\label{EstimateonR'}
\vert R'(\tm+sv(s))-R''(\tm)sv(s)-\frac{1}{2}R'''(\tm)s^2v(s)^2\vert \leq 4Bs^3v(s)^3.
\end{equation}
By the triangle inequality and part (b), 
\begin{multline*}
\left| R'(\tm+sv(s))-R''(\tm)s(v(0)+sv'(0))-\frac{1}{2}R'''(\tm)s^2v(0)^2 \right| \\
\leq 4Bv(s)^3s^3 - 
R''(\tm)((v(s)-v(0)-v'(0)s)s +\frac{1}{2}(v(s)-v(0))(v(s)+v(0))s^2\\
\leq \frac{125}{16}Bv(0)^3s^3 - 
R''(\tm)\left(\frac{5}{144}R'''(\tm)^2v(0)^7 + Bv(0)^5\right)s^3\\
+\frac{1}{2}\cdot\frac{1}{4}R'''(\tm)^2v(0)^4\cdot\frac{9}{4}v(0)s^3,
\end{multline*}
which, by \eqref{v0} and \eqref{v'0}, implies
\begin{multline}\label{EstimateonR'2}
\left| R'(\tm+sv(s))+2v(0)^{-1}s-4v(0)^{-2}v'(0)s^2\right| \\
\leq \left(\frac{157}{16}Bv(0)^3 + \frac{101}{288}R'''(\tm)^2 v(0)^5\right)s^3=:a_3s^3.
\end{multline}

To estimate the inverse of $R'$, use
\begin{multline*}
\left|\frac{1}{a_1s+a_2s^2+a_3s^3}-\frac{1}{a_1s}+\frac{a_2}{a_1^2}\right|=\left|\frac{(a_2^2-a_1a_3)s^2+a_2a_3s^3}{a_1^2(a_1s+a_2s^2+a_3s^3)} \right|\ \leq \left|\frac{6(a_2^2+\vert a_1a_3\vert)}{a_1^3}s\right|\\ 
\text{for} \ \ \vs<\min( \left|\frac{a_1}{3a_2}\right|, \left|\sqrt{\frac{a_1}{3a_3}}\right|),
\end{multline*}
which, by setting $a_1=2v(0)^{-1}$ and $a_2=4v(0)^{-2}v'(0)$, implies that
\begin{equation}\label{EstimateonDenominator}
\left| \frac{1}{R'(\tm+sv(s))}+\frac{v(0)}{2s}+v'(0)\right| \leq \frac{R'''(\tm)^2v(0)^7+18v(0)^2\vert a_3\vert}{12}\vert s \vert.
\end{equation}
Multiplying by $2\vert s\vert$ finishes the proof of (c).

(d) Differentiating \eqref{Defnofv} twice yields
\begin{equation}\label{v''Equation}
2v'(s)+sv''(s)=\frac{-2-R''(\tm+sv(s))(v(s)+sv'(s))^2}{R'(\tm+sv(s))}.
\end{equation}
From part (c) and a Taylor approximation for $R''$, 
\begin{multline*}
\left| -2-R''(\tm+sv(s))(v(s)+sv'(s))^2\right| < 999((R'''(\tm)^2v(0)^6+Bv(0)^4)s^2\\
+(R'''(\tm)^3v(0)^9+R'''(\tm)Bv(0)^7)s^3\\
+(R'''(\tm)^4v(0)^{12}+R'''(\tm)^2Bv(0)^{10}+10B^2v(0)^8)s^4\\
+(R'''(\tm)^5v(0)^{15}+R'''(\tm)^3Bv(0)^{13}+10R'''(\tm)B^2v(0)^{11})s^5\\
+(R'''(\tm)^6v(0)^{18}+R'''(\tm)^4Bv(0)^{16}+10R'''(\tm)^2B^2v(0)^{14}+10B^3v(0)^{12})s^6).
\end{multline*}
Since $\vs<(R'''(\tm)v(0)^3)^{-1},R'''(\tm)(Bv(0))^{-1}$,
\begin{equation}\label{I1}
\left| -2-R''(\tm+sv(s))(v(s)+sv'(s))^2\right|< 99999(R'''(\tm)^2v(0)^6+Bv(0)^4)s^2.
\end{equation}
By \eqref{EstimateonDenominator} and the estimates on $\vs$,
\begin{equation}\label{I2}
\left| \frac{1}{R'(\tm+sv(s))}+\frac{v(0)}{2s}+v'(0)\right| \leq 30R'''(\tm)v(0)^4.
\end{equation}
Combining \eqref{v''Equation},\eqref{I1} and \eqref{I2}, 
\begin{multline*}
\vert 2v'(s)+sv''(s)-2v'(0)\vert < 50000(R'''(\tm)^2v(0)^7+Bv(0)^5)s\\
+400000(R'''(\tm)^3v(0)^{10}+R'''(\tm)Bv(0)^8)s^2,
\end{multline*}
and using $\vs<(R'''(\tm)v(0)^3)^{-1}$ on the second term gives the result.




\end{proof}

\begin{corollary}\label{eh}
Suppose that $R$ and $\phi$ are infinitely continuously differentiable in some neighbourhood of $t_{max}$. Also suppose that $\tm$ is a local maximum of $R$ and $R''(\tm)<0$. Let $\delta_1$ and $\delta_2$ be positive numbers such that 
\[
m^2:=-R(\tm-\delta_1)=-R(\tm+\delta_2),
\]
and assume $m^2$ equals the right-hand side of \eqref{Estimateons}. Let \[
\tilde{s}=\min\left(\frac{R'''(\tm)}{50Bv(0)}, \frac{1}{50R'''(\tm)v(0)^3}\right),
\]
\[
\Lambda:=500R'''(\tm) \Vert \phi'\Vert_{\infty} v(0)^5 + 450000\Vert \phi \Vert_{\infty} (R'''(\tm)^2v(0)^7+Bv(0)^5).
\]
Then for any $N>1$,
\begin{multline*}
\left| \int_{\tm-\delta_1}^{\tm+\delta_2} e^{NR(t)}\phi(t)dt-\phi(\tm)\sqrt{\frac{-2\pi}{NR''(\tm)}}\right|\leq \\
\frac{\sqrt{\pi}}{2}\frac{\Lambda}{N^{3/2}}
+\phi(\tm)\sqrt{\frac{-2}{R''(\tm)}}\frac{e^{-N\tilde{s}}}{\sqrt{\tilde{s}}N}+e^{-N\tilde{s}^2/2}\Lambda.
\end{multline*}
\end{corollary}
\begin{proof}
Use Lemma \ref{4.8}. By part (a) of Lemma \ref{EstimateOnV}, 
\[
h(0)=\phi(\tm)\sqrt{\frac{-2}{R''(\tm)}}.
\]
By parts (c) and (d) of Lemma \ref{EstimateOnV},
\[
\sup_{0\leq\tau\leq m^2} \vert g'(\tau)\vert \leq \Lambda.
\]
When $v(0)>\sqrt{B}/R'''(\tm)$, 
\[
\frac{53}{625}\frac{R'''(\tm)}{Bv(0)}>\frac{53}{625R'''(\tm)v(0)^3}
\]
\[
\frac{1}{50\sqrt{B}v(0)^2} > \frac{1}{50R'''(\tm)v(0)^3}
\]
\[
a_3 < 16R'''(\tm)^2v(0)^5 
\]
\[
\sqrt{\frac{2v(0)^{-1}}{3a_3}}>\frac{1}{5R'''(\tm)v(0)^3},
\]
implying $m^2>(50R'''(t_{max})v(0)^3)^{-1}\geq \tilde{s}$. Similarly, when $v(0)<\sqrt{B}/R'''(\tm),$ then $m^2 > R'''(t_{max}))/(50Bv(0)) \geq \tilde{s}$. Thus $m^2>\tilde{s}$, and
\[
\int_{m^2}^{\infty} e^{-Nt}t^{-1/2}dt\leq \frac{e^{-Nm^2}}{mN}\leq\frac{e^{-N\tilde{s}}}{\sqrt{\tilde{s}}N}.
\]
\end{proof}

We can finally wrap up the proof of Proposition \ref{Prop6.9}. Since $s$ is not too small (polynomial in $v(0)^{-1}$), the exponential terms are small enough to be ignored. Therefore the error term is $\Lambda_1=\bigO(v(0)^7N^{-3/2})= \bigO(R''_1(\tm)^{-7/2}N^{-3/2})$. The main term is of order $N^{-1/2}R''_1(\tm)^{-1/2}$. Thus, when multiplying two integrals of the form in Corollary \ref{eh}, we get 
\[
\bigO\left(\frac{R_1''(t_{\text{max}_1})^{-3}+R_2''(t_{\text{max}_2})^{-3}}{R_1''(t_{\text{max}_1})^{1/2}R_2''(t_{\text{max}_2})^{1/2}}N^{-2}\right)+\bigO(R_1''(t_{\text{max}_1})^{-7/2}R_2''(t_{\text{max}_2})^{-7/2}N^{-3}),
\]
as needed.
\end{proof}

In Proposition \ref{Prop6.9}, the error term blows up at the edge. Therefore a better bound is needed. To get this bound, we simply use the first term in Watson's lemma, as opposed to using two terms. Since the method of the proof is identical as before and the details are simpler, the proof will be omitted. The exact statement is the following.

\begin{proposition}\label{EdgeAsymptotics}
For $j=1,2$, let $(\nu_j,\eta_j,\tau)\in\mathcal{D}$, $\Om_j$ denote $\Om(\nu_j,\eta_j,\tau)$, $G_j(z)$ denote $G(\nu_j,\eta_j,\tau,z)$, and $\theta_j$ denote $\theta(\nu_j,\eta_j,\tau)$. With the assumptions in section \ref{Statement of the Main Theorem}, 
\begin{multline*}
\left(\frac{1}{2\pi i}\right)^2\int_{\Gamma_1}\int_{\Gamma_2}\frac{\exp(NG(\eta_1,\nu_1,\tau,u))}{\exp(NG(\eta_2,\nu_2,\tau,w))}f(u,w)dwdu\\
\leq\frac{1000}{N\sqrt{\vert G_1''(\Om_1)\vert}\sqrt{\vert G_2''(\Om_2)\vert}}
\times \Big[\left| f(\Om_1,\Om_2)\right| + \left| f(\Om_1,\bar{\Om}_2)\right| + \left| f(\bar{\Om_1},\Om_2)\right| + \left| f(\bar{\Om_1},\bar{\Om_2})\right|\Big] 
\end{multline*}
\end{proposition}

\end{document}